\documentclass[11pt,a4]{article}

\usepackage{amsmath,amsthm,amsfonts,amssymb}
\usepackage[toc,page]{appendix}
\usepackage[a4paper,left=2cm,right=2cm,top=2cm,bottom=2cm]{geometry}
\usepackage{mathtools}
\usepackage{graphicx,tikz}
\usepackage{ifthen}
\usepackage{authblk}
\usepackage{makecell}
\usepackage{tcolorbox}
\usepackage{lipsum}
\usepackage{scalerel}
\tcbuselibrary{skins,breakable}
\usepackage{bm}
\usepackage{color}
\usepackage{authblk}

\usetikzlibrary{positioning,arrows.meta}

\newboolean{ElectronicVersion}
\setboolean{ElectronicVersion}{true}

\ifthenelse{\boolean{ElectronicVersion}}{
	\usepackage[a4paper=true,pdftex,bookmarks,pagebackref,
	plainpages=false, 
	pdfpagelabels=true 
	]{hyperref}}{}

\makeatletter
\newtheorem*{rep@theorem}{\rep@title}
\newcommand{\newreptheorem}[2]{%
	\newenvironment{rep#1}[1]{%
		\def\rep@title{#2 \ref*{##1}}%
		\begin{rep@theorem}}%
		{\end{rep@theorem}}}
\makeatother

\makeatletter
\DeclareRobustCommand{\Udots}{%
	\vcenter{\offinterlineskip
		\halign{%
			\hbox to .8em{##}\cr
			\hfil.\cr\noalign{\kern.2ex}
			\hfil.\hfil\cr\noalign{\kern.2ex}
			.\hfil\cr}%
	}%
}
\makeatother

\usepackage{hyperref}
\hypersetup{
	bookmarksnumbered=true, 
	unicode=false, 
	pdfstartview={FitH}, 
	pdftitle={The Compressed Oracle Is A Worthy (Multiplicative) Adversary}, 
	pdfauthor={Stacey Jeffery and Sebastian Zur}, 
	pdfsubject={}, 
	pdfcreator={}, 
	pdfproducer={}, 
	pdfkeywords={}, 
	pdfnewwindow=true, 
	colorlinks=true, 
	linkcolor=blue, 
	citecolor=blue, 
	filecolor=blue, 
	urlcolor=blue 
}

\newcommand{\eq}[1]{\hyperref[eq:#1]{(\ref*{eq:#1})}}
\renewcommand{\sec}[1]{\hyperref[sec:#1]{Section~\ref*{sec:#1}}}
\newcommand{\thm}[1]{\hyperref[thm:#1]{Theorem~\ref*{thm:#1}}}
\newcommand{\lem}[1]{\hyperref[lem:#1]{Lemma~\ref*{lem:#1}}}
\newcommand{\cor}[1]{\hyperref[cor:#1]{Corollary~\ref*{cor:#1}}}
\newcommand{\app}[1]{\hyperref[app:#1]{Appendix~\ref*{app:#1}}}
\newcommand{\tabl}[1]{\hyperref[tab:#1]{Table~\ref*{tab:#1}}}
\newcommand{\defin}[1]{\hyperref[def:#1]{Definition~\ref*{def:#1}}}
\newcommand{\fig}[1]{\hyperref[fig:#1]{Figure~\ref*{fig:#1}}}
\newcommand{\clm}[1]{\hyperref[clm:#1]{Claim~\ref*{clm:#1}}}
\newcommand{\conj}[1]{\hyperref[conj:#1]{Conjecture~\ref*{conj:#1}}}
\newcommand{\rem}[1]{\hyperref[rem:#1]{Remark~\ref*{rem:#1}}}
\newcommand{\para}[1]{\hyperref[para:#1]{Paragraph~\ref*{para:#1}}}
\newcommand{\exmp}[1]{\hyperref[exmp:#1]{Example~\ref*{exmp:#1}}}
\newcommand{\appx}[1]{\hyperref[appx:#1]{Appendix~\ref*{appx:#1}}}
\newcommand{\fct}[1]{\hyperref[fct:#1]{Fact~\ref*{fct:#1}}}

\newcommand{\thmthm}[2]{\hyperref[thm:#1]{Theorem~\ref*{thm:#1}} and~\hyperref[thm:#2]{\ref*{thm:#2}}}
\newcommand{\lemlem}[2]{\hyperref[lem:#1]{Lemma~\ref*{lem:#1}} and~\hyperref[lem:#2]{\ref*{lem:#2}}}

{\endtcolorbox}

\newcommand{\nocontentsline}[3]{}
\newcommand{\tocless}[2]{\bgroup\let\addcontentsline=\nocontentsline#1{#2}\egroup}

\newtheorem{theorem}{Theorem}[section]
\newtheorem{lemma}[theorem]{Lemma}

\newtheorem{corollary}[theorem]{Corollary}
\newtheorem{fact}[theorem]{Fact}
\newtheorem{claim}[theorem]{Claim}

\newtheorem{remark}[theorem]{Remark}
\newtheorem{definition}[theorem]{Definition}

\usepackage{color}
\definecolor{darkgreen}{rgb}{0,.5,0}
\definecolor{darkred}{rgb}{.7,.3,.3}
\definecolor{deepblue}{rgb}{0,.1,.7}

\newif\iflongversion
\newif\ifshortversion
\newif\ifediting
\editingtrue 

\def\ket#1{{\lvert}#1\rangle}
\def\bra#1{{\langle}#1\rvert}

\def\braket#1#2{{{\langle}#1\vert}#2\rangle}
\def\abs#1{\left| #1 \right|}

\def\norm#1{\left\| #1 \right\|}

\def\Tr{\mbox{Tr}}

\def\Func{{\sf Func}}
\def\proj#1{\ket{#1}\bra{#1}}

\tikzset{
	mybox/.style = {
		rectangle,
		rounded corners=3mm,
		draw=blue!70!black,
		thick,
		fill=blue!5,
		minimum width=3cm,
		minimum height=1cm,
		align=center
	},
	arr/.style = {
		-{Straight Barb[scale=1]},
		line width=1pt,
		draw=blue!70!black
	}
}

\title{The Compressed Oracle is a Worthy (Multiplicative) Adversary}
\author[1]{Stacey Jeffery\thanks{This work is supported by ERC STG grant 101040624-ASC-Q and NWO Klein project number OCENW.Klein.061. SJ is a CIFAR Fellow in the Quantum Information Science Program.}}
\author[2]{Sebastian Zur\thanks{The majority of this work was conducted while SZ was affiliated with CWI \& QuSoft, the Netherlands.}}
\affil[1]{QuSoft, CWI \& University of Amsterdam, the Netherlands}
\affil[2]{IRIF \& CNRS, France}

\begin{document}
	
\maketitle

\begin{abstract} 
The compressed oracle technique, introduced in the context of quantum cryptanalysis, is the latest method for proving quantum query lower bounds, and has had an impressive number of applications since its introduction, due in part to the ease of importing classical lower bound intuition into the quantum setting via this method. Previously, the main quantum query lower bound methods were the polynomial method, the adversary method, and the multiplicative adversary method, and their relative powers were well understood. In this work, we situate the compressed oracle technique within this established landscape, by showing that it is a special case of the multiplicative adversary method. To accomplish this, we introduce a simplified restriction of the multiplicative adversary method, the \emph{MLADV} method, that remains powerful enough to capture the polynomial method and exhibit a strong direct product theorem, but is much simpler to reason about. We show that the compressed oracle technique is also captured by the MLADV method. This might make the MLADV method a promising direction in the current quest to extend the compressed oracle technique to non-product distributions.
\end{abstract}

\section{Introduction}

Proving quantum query lower bounds is essential to understanding the limitations of quantum computers. 
In the \emph{bounded-error quantum query model}, an algorithm for a problem ${\sf F}$ receives its input -- typically a string in $[M]^N$ for some integers $M$ and $N$ -- encoded as a function $f:[N]\rightarrow [M]$, accessible only through queries. The algorithm may alternate such queries with arbitrary quantum operations, and it must produce the correct output on every input with probability at least $2/3$.   
The \emph{bounded-error quantum query complexity} of ${\sf F}$, denoted ${Q}({\sf F})$, is the minimum number of queries needed by any such algorithm. Allowing some small probability of error makes this a practical model of computation, and lower bounds on the query complexity of an algorithm are also lower bounds on the total number of steps the algorithm must make. 

The first technique for proving quantum query lower bounds was the \emph{polynomial method}~\cite{beals2001QLowerBoundPoly}, which showed that the acceptance probability of a quantum algorithm can be represented by a low-degree polynomial. In this method, one lower bounds the quantum query complexity of ${\sf F}$ by lower bounding its \emph{approximate degree} $\widetilde{\deg}({\sf F})$, by 
proving a lower bound on the degree of any polynomial with certain properties that must be satisfied by any successful algorithm. Later, the \emph{adversary method} was introduced in~\cite{ambainis2002PositiveAdv} and generalized to its full version in~\cite{hoyer2007NegativeAdv}. Letting ${\sf ADV}^\pm({\sf F})$ denote the best possible lower bound on the quantum query complexity of ${\sf F}$ that one can prove using the adversary method, it was later shown that this quantity is equal, up to constants, to ${Q}({\sf F})$, making this a very powerful method. In contrast, there are problems for which the polynomial method is not able to prove tight lower bounds, since $\widetilde{\deg}({\sf F})=o({Q}({\sf F}))$~\cite{aaronson2015cheatSheets}. However, the polynomial method has an advantage over the adversary method. While the adversary method is only able to prove non-trivial lower bounds on \emph{bounded-error} quantum query complexity, the polynomial method can be used to prove lower bounds on ${Q}_{\epsilon}({\sf F})$, the minimum number of queries needed by any quantum algorithm to compute ${\sf F}$ with success probability at least $1-\epsilon$, even when $1-\epsilon=o(1)$. This is particularly useful in cryptographic settings, as we discuss shortly. 

A later more powerful variant of the adversary method is the \emph{multiplicative adversary method}~\cite{vspalek2008multiplicative}, which improves on the adversary method by allowing for lower bounds on ${Q}_{\epsilon}({\sf F})$ even when the success probability $1-\epsilon$ is very small. This method is at least as powerful as both the adversary method and the polynomial method, but it has few applications simply because it is very difficult to apply. In both lower bound techniques and algorithmic techniques, there is often a tradeoff between the power of a technique, and its ease of application, and an important pursuit is to find techniques with just the right balance of power and ease of use. 

\begin{figure}
	\centering
	\begin{tikzpicture}
		\node[mybox] (adv)    at (-4,  0) {${\sf ADV}_{\epsilon}({\sf F})$};
		\node[mybox] (advpm)  at (-4,  2) {${\sf ADV}_{\epsilon}^{\pm}({\sf F})$};
		\node[mybox] (madv)   at (2,  4) {${\sf MADV}_{\epsilon}({\sf F})$};
		\node[mybox] (mladv)  at (2,  2) {${\sf MLADV}_{\epsilon}({\sf F})$};
		\node[mybox] (deg)    at (0,  0) {$\widetilde{\deg}_\epsilon({\sf F})$};
		\node[mybox] (comp)   at (4,  0) {${\sf COMP}_{\epsilon}({\sf F})$};
		\draw[arr] (adv)   -- node[left]  {\textcircled{1}} (advpm);
		\draw[arr] (advpm) -- node[pos=0.4,above] {\textcircled{2}} (madv);
		\draw[arr] (deg) -- node[pos=0.4,above] {\textcircled{3}} (advpm);
		\draw[arr] (mladv) -- node[left] {\hyperref[sec:mladv]{Sec.~\ref*{sec:mladv}}} (madv);
		\draw[arr] (deg)   -- node[pos=0.2,above,xshift=-.4cm] {\hyperref[sec:reduction-poly]{Sec.~\ref*{sec:reduction-poly}}} (mladv);
		\draw[arr] (comp)  -- node[pos=0.2,above,xshift=.4cm] {\hyperref[sec:reduction]{Sec.~\ref*{sec:reduction}}} (mladv);
		\node             (ineq) at (-2,0) {$\lessgtr$};
		\node at (-2,0.4) {\textcircled{4}} ;
	\end{tikzpicture}
	\caption{The relationships between the various methods to obtain quantum query lower bounds, expanding on a similar figure in ~\cite{magnin2015explicit}. An arrow from method A to method B implies that for any lower bound that can be proven with A, we can explicitly construct a lower bound with B (i.e., B is stronger than A). \textcircled{1}~\cite{hoyer2007NegativeAdv};~\textcircled{2}~\cite{ambainis2011symmetry};~\textcircled{3}~\cite{belovs2024direct} only holds in the bounded-error regime;~\textcircled{4} The original additive and the polynomial methods are incomparable~\cite{zhang2005power,ambainis2006polynomial}.
	Technically, there are two slightly different definitions of ${\sf MLADV}_\epsilon({\sf F})$ defined in this work: a simpler to state one as in~\thm{reduction}, and a stronger one in~\thm{reduction-poly}. This mirrors the situation with ${\sf MADV}_\epsilon({\sf F})$, which can denote the slightly weaker bound from~\cite{vspalek2008multiplicative}, or the stronger variant from~\cite{lee2013strong} that is slightly more complicated to state, though no more complicated to apply. 
	In this figure, we mean the stronger version of both.
	}\label{fig:overview}
\end{figure}
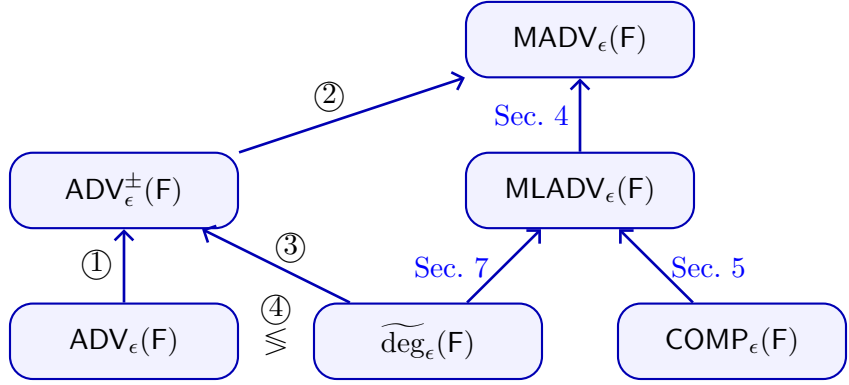

A relative newcomer to the landscape of quantum query lower bound techniques is the \emph{compressed oracle technique}. This was introduced in~\cite{zhandry2019record}, and distilled into a formal framework in~\cite{chung2021compressed}. The compressed oracle framework was introduced in the context of post-quantum cryptanalysis, in order to lower bound the work needed by a quantum \emph{adversary} that interacts with a \emph{quantum random oracle} -- a uniform random function $f:X\rightarrow Y$ that the adversary can query in superposition. As such an adversary's goal is generally something nefarious, it is critical to show the impossibility of efficient adversaries that achieve their goal, even with success probabilities below a constant. This technique has received widespread use since its introduction, and seems to have a particularly nice balance of power and ease-of-use. 
There is a nice intuition behind the technique that makes it particularly well suited for adapting classical intuitions about the hardness of a problem to the quantum world, but it has also been powerful enough to prove a number of new results. Still, its use has remained mostly restricted to the setting of uniform functions, and it has proven resistant to relatively minor modifications, such as a generalisation to non-product distributions. 

While the relationship between the other mentioned methods is well established, prior to this work, it was unknown where the compressed oracle method fit into the picture.

\paragraph{Contributions:} In this work, we show that the compressed oracle technique can be viewed as a special case of the multiplicative adversary method, fitting it into the existing landscape of quantum query lower bound techniques. To do this, we introduce a simplification of the multiplicative adversary method called the \emph{multiplicative ladder adversary method} (${\sf MLADV}_{\epsilon}$, \sec{mladv}), and show that the compressed oracle technique ${\sf COMP}_{\epsilon}$ is actually a special case of this restricted method (\sec{reduction}). This restriction of the multiplicative adversary introduces intuitive structure that makes it, in principle, easier to apply, but we show that it is still powerful enough to capture the polynomial method (\sec{reduction-poly}). 
In fact, to the best of our knowledge, virtually every use of the multiplicative adversary to date is an instance of a multiplicative ladder adversary. 
As further evidence of its natural structure, we show (\sec{dpt}) that ${\sf MLADV}_{\epsilon}$ exhibits a strong direct product theorem. 
These relationships (and others) are summarized in \fig{overview}.

\vskip10pt
We now give a more detailed survey of quantum query lower bound techniques, and discussion of our results. 

\subsection{Adversary methods}

The original adversary method ${\sf ADV}_{\epsilon}$ for proving quantum query lower bounds was first introduced in~\cite{ambainis2002PositiveAdv}. Applying this method reduces to mostly combinatorial arguments, which makes it very convenient to use, as shown by its many applications~\cite{barnum2004lower, durr2004QQueryCompGraph, buhrman2006MatrixProduct, dorn2007quantum}. However, this method does have some technical limitations, one of which is the \textit{certificate complexity barrier}~\cite{zhang2005power}, which shows that there are problems for which this method cannot be tight. This limitation is addressed by the strictly stronger negative-weights adversary method ${\sf ADV}_{\epsilon}^{\pm}$ by~\cite{hoyer2007NegativeAdv}, now usually just referred to as \emph{the adversary method}. This method is capable of proving tight lower bounds on the quantum query complexity of any $\sf F$ in the bounded-error regime~\cite{reichardt2009span}, but this power comes at the cost of making it more complicated to apply, as its greater abstraction removes the primarily combinatorial reasoning suggested by the constraints of the original adversary method. This means that even for very symmetric problems such as the \textit{collision problem}~\cite{aaronson2004QLowerBndCollisionAndElementDistinct}, it is highly difficult to come up with a non-trivial lower bound using the adversary method, and the only known construction relies on studying the symmetries of the problem via representation theory~\cite{belovs2013adversary}. Lower bounds on ${Q}_\epsilon({\sf F})$ proven using the adversary method are proportional to $1-\epsilon$, the algorithm's success probability, making them negligible for exponentially small success probabilities. It is therefore only suitable for proving lower bounds in the bounded-error regime.

The latest and most powerful iteration in adversary methods, and the one most central to this work, is the multiplicative adversary method ${\sf MADV}_{\epsilon}$ formalized in~\cite{vspalek2008multiplicative,lee2013strong}, as a generalisation of an ad-hoc technique proposed in~\cite{ambainis2006new,ambainis2005new}. This method is shown to be strictly stronger than the adversary method~\cite{ambainis2011symmetry}. Since the adversary method is already tight in the bounded-error regime, this generalisation is particularly relevant in the low success probability regime, where it works even for exponentially small probabilities of success. This is a necessary condition for the method to exhibit a \textit{strong direct product theorem} (SDPT), which intuitively states that to solve $k$ independent instances of a function, one needs $\Omega(k)$ times as many queries to achieve even an exponentially small (in $k$) probability of success. It was already shown in~\cite{vspalek2008multiplicative} that the multiplicative adversary method satisfies a SDPT, which allowed~\cite{lee2013strong} to prove a SDPT for quantum query complexity. 
The multiplicative adversary's applicability to the small success probability regime has also proven useful in proving quantum time-space tradeoff lower bounds~\cite{ambainis2006new}.
However, just as the (negative-weights) adversary method is more complicated to apply than the original adversary method, the powerful multiplicative adversary method is even more complicated and, as a result, still has relatively few applications.

\subsection{Polynomial method}

Another technique for proving quantum query lower bounds is the polynomial method~\cite{beals2001QLowerBoundPoly}, which predates the adversary method. The method relies on the principle that for any $T$-query quantum algorithm, its acceptance probability can be expressed as a multivariate polynomial of degree $2T$ in the input variables. In any algorithm that computes ${\sf F}$ with error $\epsilon$, this polynomial gives an $\epsilon$-approximation to ${\sf F}$, and so its degree, $2T$, is at least $\widetilde{\deg}_\epsilon({\sf F})$, the minimum degree of any polynomial that $\epsilon$-approximates ${\sf F}$. This allows one to prove lower bounds on ${Q}_\epsilon({\sf F})$ using results about polynomials. For example, the fact that a polynomial that changes value many times must have high degree implies a lower bound on problems like parity, that change value many times.  
Using much more involved reasoning, the polynomial method was used to give a tight lower bound on the collision problem~\cite{aaronson2004QLowerBndCollisionAndElementDistinct}, whereas an adversary lower bound for this problem was only constructed much later~\cite{belovs2013adversary}. 

While the polynomial method is incomparable to the original adversary method~\cite{zhang2005power,ambainis2006polynomial}, the (negative-weights) adversary method subsumes it in the bounded-error regime; a reduction recently made constructive by~\cite{belovs2024direct}. The polynomial method does, however, work for small success probabilities, which makes it possible to prove SDPTs~\cite{klauck2007quantum,sherstov2011strong}. Furthermore, as shown by~\cite{magnin2015explicit}, it can be reduced to the multiplicative adversary method.

\subsection{Compressed oracle technique}

For cryptographic security proofs, lower bounds on bounded-error quantum query complexity make little sense. Ruling out adversaries that succeed with a high probability of success (at least $2/3$) is not enough, and it is necessary to rule out adversaries with very small success probabilities as well. Moreover, worst-case query complexity, as captured by ${Q}_\epsilon({\sf F})$, is not the relevant quantity. Imagine an adversary whose task ${\sf F}$ is to break some cryptosystem using its public key as input. A lower bound on ${Q}_\epsilon({\sf F})$ only proves that there is \emph{some} key on which the adversary requires many queries; it says nothing about the adversary's resource requirements on a \emph{random} key.  This is instead captured by \emph{average-case} quantum query complexity, which is defined with respect to some distribution of interest (often the uniform distribution). 

The compressed oracle technique~\cite{zhandry2019record}, introduced in the context of cryptographic security proofs, does precisely this, as it yields an upper bound on the probability of success for any quantum algorithm interacting with a random oracle, giving an average-case lower bound\footnote{Adversary methods also work by giving an average-case lower bound with respect to some distribution of inputs, which then implies a lower bound on the worst-case complexity. However, as the goal in using these is usually a worst-case lower bound, generally a deliberately hard distribution is chosen, rather than a uniform one.} that holds even for exponentially small probabilities of success. Moreover, its analysis works via mostly combinatorial arguments that look quite similar to the types of reasoning one would use to prove classical lower bounds, which makes it straightforward to apply and has quickly resulted in many results~\cite{liu2019finding,czajkowski2019quantum,liu2019revisiting,chung2021compressed,grilo2021tight,don2022online}. It also satisfies a SDPT, and has even been used to prove quantum time-space tradeoffs~\cite{hamoudi2020quantum}. The limitation of this technique, however, is that it is not known how to apply it on input distributions where the values $f(x)$ for different $x$ are not independent~\cite{czajkowski2019quantum,hamoudi2020quantum}. This limitation 
makes it difficult to prove worst-case lower bounds, as the hardest input distributions often have some global structure. It also rules out applications involving certain interesting cryptographic primitives such as random permutations. For that specific case, an ad-hoc workaround has recently been devised by~\cite{alagic2025sponge}, extending the indifferentiability of the Sponge construction~\cite{bertoni2007sponge} to the quantum setting. From the standpoint of quantum query lower bounds, however, this remedy is no longer tight.

\subsection{${\sf COMP}_{\epsilon}$ vs.~other techniques}

All methods discussed above operate by tracking some progress measure that must change significantly over the course of the algorithm, but can only change a small amount using a single query. For the polynomial method, this is the degree of the polynomials representing the amplitudes of the algorithm's states. For the compressed oracle and the adversary methods, the measure of progress is somehow measuring the amount of entanglement between the algorithm and the input, when instantiated as a coherent superposition. Since the adversary and compressed oracle techniques have different drawbacks that do not seem to exist in the other, it is interesting to see what the explicit relationship between these techniques is. This could aid in the ongoing search for a fusion of both techniques: a compressed oracle technique that can be applied to input distributions where each $f(x)$ is not necessarily assigned independently. On the cryptographic side, this could lead to (better) quantum security proofs for schemes using random permutations, such as the sponge construction~\cite{bertoni2007sponge}. On the quantum query lower bounds side, this might result in a technique that marries the power of the multiplicative adversary method — which works for all input distributions — with the intuitive combinatorial reasoning of the compressed oracle technique. Currently, the most promising result towards this ``holy grail'' has been a representation theory approach by~\cite{rosmanis2021tight} that allows for tackling the problem of inverting a random permutation.

In this work, we demonstrate that a generalised compressed oracle technique — one that accommodates distributions beyond random functions and permutations — must fall somewhere between the compressed oracle technique and the multiplicative adversary method. We explicitly show this by proving that the compressed oracle technique reduces to the multiplicative adversary method. We achieve this by defining a weaker version of the multiplicative adversary method, the multiplicative ladder adversary ${\sf MLADV}_{\epsilon}$. An adversary lower bound (multiplicative or standard) is proven by exhibiting an \emph{adversary matrix} that satisfies certain properties. In the ${\sf MLADV}_{\epsilon}$ technique, we restrict adversary matrices to those whose eigenvalues are increasing powers of some constant larger than $1$, and whose eigenspaces form a ``ladder'' in the sense that a query can move the state up or down at most one eigenspace. This ladder structure makes reasoning about an algorithm's progress much more tractable. 

The ${\sf MLADV}_{\epsilon}$ method still satisfies a strong direct product theorem (SDPT, see \sec{dpt}) and remains more powerful than the compressed oracle technique. Additionally, we show that this new version also still encompasses the polynomial method (see \sec{reduction-poly}). These results are summarized in \fig{overview}. We hope that this new intermediate technique will aid in the search for an extended compressed oracle technique, as we show that it incorporates the approach from~\cite{rosmanis2021tight} to random permutations as a special case.

\section{Preliminaries}

\subsection{Linear algebra}

In this work we consider finite-dimensional complex inner product spaces ${\cal H} = \mathbb{C}^d$ for some dimension $d$. We use standard bra-ket notation for column and row vectors in $\mathbb{C}^d$.  We consider all bra-ket vectors to be normalised unless specified otherwise. For a finite set $S$, we let
$$\mathbb{C}^S = \mathbb{C}[S] = \mathrm{span}\{\ket{s}:s\in S\},$$
using whichever notation is most convenient given the complexity of writing $S$. 
For any two Hermitian operators $A, B$, we write $A\succeq B$ if their difference $A - B$ is positive semidefinite.
\begin{definition}[Spectral norm]
	Let $A \in \mathbb{C}^{d \times d}$ be a matrix. Then the \emph{spectral norm} (also known as the operator norm) of $A$ is 
	$$\norm{A} := \sup\limits_{\ket{v} \in \mathbb{C}^{d}} \norm{A\ket{v}},$$
	where $\norm{A\ket{v}}$ is the standard vector $\ell_2$-norm.
\end{definition}
\noindent We will make use of the following standard result.
\begin{lemma}\label{lem:spectral}
For any linear operator $A$, the spectral norm of $A$ satisfies
$$\norm{A}\leq \sqrt{\norm{A}_1\norm{A}_\infty}.$$
\end{lemma}

\subsection{Quantum query complexity}\label{sec:query}

In the quantum query model, we are generally interested in computing a function ${\sf F}:\Func\rightarrow \Sigma$ on an input $f \in \Func$. We consider the case where $\Func$ is a subset of $Y^X$, so each $f$ can itself also be viewed as a function from $X$ to $Y$. 
For example, if $Y = \{0,1\}$ and $X = [n] := \{1,\dots,n\}$, then $f \in \Func$ is an $n$-bit string (which might have a promise defined by the subset $\Func$). In this work, we usually restrict ourselves to $X$ being any finite set of size $N$ and consider $Y$ to be the finite set $[M-1]_0 := \{0,\dots,M-1\}$.

The memory of our quantum algorithm ${\cal A}$, tasked with computing ${\sf F}$ on an input $f$, is described without loss of generality by the registers ${\cal W}$, ${\cal X}$, and ${\cal Y}$. Here, the input oracle acts on ${\cal X} \times {\cal Y}$ (as detailed below), while ${\cal W}$ represents an additional workspace. The input function $f \in \Func$ can be accessed by ${\cal A}$ via an oracle, defined as follows:

\begin{definition}[Oracle]\label{def:query-og}
	Fix a finite set $X$ of size $N$ and let $Y = [M-1]_0$. An \emph{oracle} ${\cal O}_f$, encoding the input function $f \in \Func$, is a unitary transformation that acts on
	$$ \mathrm{span}\{\ket{x}_{\cal X}\ket{y}_{\cal Y} : x \in X, y \in Y\}, $$
	with its action on the basis state $\ket{x}_{\cal X}\ket{y}_{\cal Y}$ defined as
	$$ {\cal O}_f\ket{x}_{\cal X}\ket{y}_{\cal Y} = \ket{x}_{\cal X}\ket{(y + f(x)) \mathrm{~mod~} M}_{\cal Y}. $$
\end{definition}

The input $f$ is typically drawn from some (hard) input distribution $\delta$ over $\Func$, denoted $f \sim \delta$. Consequently, ${\cal O}_f$ is a random variable. In adversary methods and the compressed oracle technique, this randomness is avoided by introducing an additional \emph{input} register $\mathcal{I}$, which stores a superposition of function tables representing the input $f$. In quantum information theory, this is known as \textit{purification}. If $f \sim \delta$, the register $\mathcal{I}$ will be initialised as
$$
\ket{\delta} = \sum_{f \in Y^X} \sqrt{\delta(f)} \ket{f}_{\cal I}.
$$
Here, $\ket{\delta}$ represents the initial state of the input register. It is important to note that this should not be confused with the initial state of the algorithm, which is the all-zero state. This purification of the input leads to the following purified oracle:
\begin{definition}[Purified Oracle]\label{def:query}
	Fix a finite set $X$ of size $N$ and let $Y = [M-1]_0$. A \emph{purified oracle} ${\cal O}$ is a unitary transformation that acts on
	$$ \mathrm{span}\{\ket{x}_{\cal X}\ket{y}_{\cal Y}\ket{f}_{\cal I} : x \in X, y \in Y, f \in Y^X\}, $$
	with its action on the basis state $\ket{x}_{\cal X}\ket{y}_{\cal Y}\ket{f}_{\cal I}$ defined as
	$$ {\cal O}\ket{x}_{\cal X}\ket{y}_{\cal Y}\ket{f}_{\cal I} = \ket{x}_{\cal X}\ket{(y + f(x)) \mathrm{~mod~} M}_{\cal Y}\ket{f}_{\cal I}. $$
\end{definition}

From the perspective of the algorithm, it is indistinguishable whether it interacts with the random variable ${\cal O}_f$ or the purified oracle ${\cal O}$ with input register initialised to $\ket{\delta}$. The relationship between the two is captured by the following expression:
$$
{\cal O} = \sum_{f \in Y^X} {\cal O}_f \otimes \proj{f}_{\cal I}.
$$

It is equivalent, and in this work more convenient, to encode the query into the phase by viewing the ${\cal Y}$ register in the Fourier basis $\{\ket{\hat{y}}\}_{y \in Y}$ instead of the computational basis $\{\ket{y}\}_{y \in Y}$.

\begin{definition}[Fourier basis]\label{def:fourier}
	Let $Y = [M-1]_0$ and let $\{\ket{y}\}_{y \in Y}$ be the computational basis for ${\cal Y} = \mathbb{C}^M$. Then $\{\ket{\hat{y}}\}_{y \in Y}$ is the \emph{Fourier basis} of ${\cal Y}$, where each $\ket{\hat{y}}$ is defined as
	$$ \ket{\hat{y}} = \frac{1}{\sqrt{M}} \sum_{z \in Y} e^{\frac{2\pi\iota}{M}y z}\ket{z}.$$
	Here $\iota$ denotes the imaginary unit to prevent ambiguity with the variable $i$. The unitary map $\ket{y} \mapsto \ket{\hat{y}}$ is also known as the \emph{Quantum Fourier Transform over the integers mod $M$}, which we denote ${\sf QFT}_M$.	
\end{definition}

\noindent In this Fourier basis, the oracle from \defin{query} acts on any basis state $\ket{x}_{\cal X}\ket{\hat{y}}_{\cal Y}\ket{f}_{\cal I}$ as
$$ {\cal O}\ket{x}_{\cal X}\ket{\hat{y}}_{\cal Y}\ket{f}_{\cal I} = e^{\frac{2\pi\iota}{M}y f(x)}\ket{x}_{\cal X}\ket{\hat{y}}_{\cal Y}\ket{f}_{\cal I}. $$

Additionally, it will often be convenient to decompose the oracle ${\cal O}$ into diagonal unitary matrices ${\cal O}_{x,y}$ given by
\begin{equation}\label{eq:query}
	{\cal O} = \sum_{x \in X, y \in Y} \proj{x}_{\cal X} \otimes \proj{\hat{y}}_{\cal Y} \otimes {\cal O}_{x,y},
\end{equation}
where each ${\cal O}_{x,y}$ acts on the basis state $\ket{f}_{\cal I}$ as
$$ {\cal O}_{x,y}\ket{f}_{\cal I} = e^{\frac{2\pi\iota}{M}y \cdot f(x)}\ket{f}_{\cal I}. $$

\begin{definition}[$T$-Query Quantum Algorithm]\label{def:algorithm}
	Fix a set $X$ of size $N$ and let $Y = [M-1]_0$. A \emph{$T$-query quantum algorithm} ${\cal A}$ on $Y^X$ is a sequence of unitaries $U_0,\dots,U_T$ on
	$$ \mathrm{span}\{\ket{w}_{\cal W}\ket{x}_{\cal X}\ket{y}_{\cal Y}:w \in W, x \in X, y \in Y\}, $$
	for some finite set $W$. For a fixed algorithm ${\cal A}$ and a fixed input distribution $\delta$, let
	$$ \ket{\delta} = \sum_{f \in Y^X} \sqrt{\delta(f)}\ket{f}_{\cal I}, $$
	and let 
	$$ \ket{\psi_t({\cal A},\delta)} = U_t{\cal O}U_{t-1}{\cal O}\dots {\cal O}U_0\ket{0}_{\cal WXY}\ket{\delta}_{\cal I} $$
	denote the state of the algorithm before the $(t+1)$-th query is made, and let
	$$ \rho_{\cal I}^t({\cal A},\delta) = \emph{\Tr}_{\cal WXY}\left[\ket{\psi_t({\cal A},\delta)}\bra{\psi_t({\cal A},\delta)}\right] $$
	denote the reduced state of the input register, which we call the \emph{input register states for ${\cal A}$ and $\ket{\delta}$}. When ${\cal A}$ and $\ket{\delta}$ are clear from context, we will omit the $({\cal A},\delta)$ notation.
\end{definition}
In the definition of $\ket{\psi_t({\cal A},\delta)}$, both the queries ${\cal O}$ and the unitaries $U_1,\dots,U_t$ act on a larger Hilbert space than originally defined, but each operator is implicitly understood to act tensored with the identity operator on any unaffected registers.

In this work, we compare various techniques designed to lower bound the quantum query complexity of a problem ${\sf F}$:
\begin{definition}[$\epsilon$-error Quantum Query Complexity]\label{def:complex}
	Fix ${\sf F}:\Func\rightarrow \Sigma$. Then the \emph{$\epsilon$-error quantum query complexity} of ${\sf F}$, denoted by $Q_{\epsilon}({\sf F})$, is the minimum number of queries needed by any quantum query algorithm ${\cal A}$ to successfully output ${\sf F}(f)$ for every input $f \in \Func$ with success probability at least $1 - \epsilon$.
\end{definition}

\section{The frameworks}\label{sec:fwks}

In this section, we introduce the two main lower bound frameworks that will be compared throughout this work: the multiplicative adversary method and the compressed oracle technique. The other lower bound method discussed in this paper, the polynomial method, is not needed until \sec{reduction-poly}, and we define it there.

\subsection{The multiplicative adversary method}

The general idea behind the adversary methods is that any algorithm for ${\sf F}$, run on a superposition of different inputs $\ket{\delta}$ with different values of ${\sf F}$, must entangle the algorithm's workspace ${\cal WXY}$ (which must eventually contain the answer) with the input register ${\cal I}$, resulting in the reduced density matrix on ${\cal I}$, which is initially the pure state $\rho_{\cal I}^0({\cal A},\delta) = \ket{\delta}\bra{\delta}$, becoming some mixed state $\rho^T_{\cal I}({\cal A},\delta)$.

This idea was already present in the original \textit{quantum adversary method}~\cite{ambainis2002PositiveAdv}, which was later generalised to the stronger \textit{negative-weights adversary method}~\cite{hoyer2007NegativeAdv} (now often called the adversary method), which is tight in the bounded-error regime, i.e.~$\epsilon \leq 1/3$. We will be interested in the even more powerful \emph{multiplicative adversary method}, first formalised in~\cite{vspalek2008multiplicative} and further developed in~\cite{ambainis2011symmetry,lee2013strong,magnin2015explicit}. We now describe this method.

\begin{definition}[Multiplicative Adversary Matrix]\label{def:adv-matrix}
	Fix ${\sf F}:\Func\rightarrow \Sigma$. A \emph{multiplicative adversary matrix} for problem ${\sf F}$ is a positive definite matrix $\Gamma\in \mathbb{C}^{\Func\times \Func}$ with smallest eigenvalue 1.
\end{definition}

Any multiplicative adversary matrix gives rise to a \emph{progress measure}, which is a way of quantifying how much progress a quantum algorithm ${\cal A}$ has made after $t$ queries towards solving a particular problem~${\sf F}$.

\begin{definition}[Progress]\label{def:progress}
	Fix a problem ${\sf F}:\Func\rightarrow\Sigma$, and input distribution $\delta$ supported on $\Func$. Fix a multiplicative adversary matrix $\Gamma$ for ${\sf F}$, as in \defin{adv-matrix}, with eigenstate $\ket{\delta}$ and a $T$-query quantum algorithm ${\cal A}$, as in \defin{algorithm}. Let $\rho_{\cal I}^t({\cal A},\delta)$ be the input register states for ${\cal A}$ and input distribution $\delta$ before the $(t+1)$-th query is made. The associated \emph{progress measure} for $t\in[T]_0$ is defined as
	\begin{equation*}
		W^t(\Gamma,{\cal A}) \coloneqq \emph{\Tr}[\Gamma \rho_{\cal I}^t({\cal A},\delta)].
	\end{equation*}
\end{definition}

\thm{madv} quantifies in what way we can think of $W^t(\Gamma,{\cal A})$ as a ``progress measure.'' After 0 queries, we have made no progress, which is indicated by $W^0(\Gamma,{\cal A})=1$ (Item 1). After $T$ queries, if we want to claim that the algorithm actually solves ${\sf F}$ with probability $1-\epsilon$, then it must be the case that the progress $W^T(\Gamma,{\cal A})$ has increased sufficiently above 1 (Item 3). Item 2 bounds the amount of progress that can be made in a single query.

\begin{theorem}[{\cite{vspalek2008multiplicative, ambainis2011symmetry}}]\label{thm:madv}
	Fix a problem ${\sf F}:\Func\rightarrow\Sigma$, an input distribution $\delta$ on $\Func$, and a multiplicative adversary matrix $\Gamma$ for ${\sf F}$ with 1-eigenstate $\ket{\delta}$. Let $\lambda$ be a real number with $1<\lambda\leq \norm{\Gamma}$. Let $\Lambda_{{\sf bad}}$ be the projector onto the eigenspaces of $\Gamma$ corresponding to eigenvalues smaller than $\lambda$ and let $\eta\leq 1-\epsilon$ be a positive constant such that $\norm{F_z\Lambda_{{\sf bad}}}^2\leq \eta$ for every $z\in \Sigma$, where $F_z = \sum\limits_{\substack{f \in \Func:\\{\sf F}(f) = z}}\proj{f}$. Then:
	\begin{enumerate}
		\item For any quantum algorithm ${\cal A}$, $W^0(\Gamma,{\cal A})=1$.
		\item For any $T$-query quantum algorithm ${\cal A}$, and $t\in[T-1]_0$, 
		$$\displaystyle\frac{W^{t+1}(\Gamma,{\cal A})}{W^t(\Gamma,{\cal A})}\leq \max\limits_{x\in X,y\in Y}\norm{{\cal O}_{x,y}^{\dagger}\Gamma^{1/2} {\cal O}_{x,y}\Gamma^{-1/2}}^2.$$
		\item For any $T$-query quantum algorithm ${\cal A}$ that solves ${\sf F}$ on input $\ket{\delta}$ with success probability at least $1-\epsilon$, $W^T(\Gamma,{\cal A})\geq 1 + \left(\lambda - 1\right)\left(\sqrt{1-\epsilon} - \sqrt{\eta}\right)^2$.
	\end{enumerate}
\end{theorem}

\begin{corollary}\label{cor:madv}
	For any $\eta$ that satisfies the constraints of \thm{madv}, $\epsilon\in (0,1-\eta)$, problem ${\sf F}:\Func\rightarrow\Sigma$, and input distribution $\delta$ on $\Func$,
	\begin{equation*}
		Q_{\epsilon}({\sf F})\geq \max\limits_{\Gamma,\lambda}\frac{\log\left(1 + \left(\lambda - 1\right)\left(\sqrt{1-\epsilon} - \sqrt{\eta}\right)^2\right)}{\log\left(\max\limits_{x \in X,y \in Y}\norm{{\cal O}_{x,y}^{\dagger}\Gamma^{1/2} {\cal O}_{x,y}\Gamma^{-1/2}}^2\right)},
	\end{equation*}
	where $\Gamma$ ranges over all multiplicative adversary matrices for ${\sf F}$ with $1$-eigenstate $\ket{\delta}$ (see \defin{adv-matrix}) and $\lambda$ ranges over $[1, \norm{\Gamma}]$.
\end{corollary}

\subsection{Dealing with search problems}\label{sec:search}

By \defin{complex}, we aim to lower bound the number of queries that any quantum query algorithm makes to successfully output ${\sf F}(f) \in \Sigma$ for any input $f \in \Func$. All \textit{decision} problems can be phrased in this form, where the set $\Sigma$ is equal to $\{0,1\}$. However, it is not always possible to interpret more general \textit{search} problems as computing a single-valued function ${\sf F}(f)$. 

For instance, consider the simplest search problem, known as \textit{Search}. If we restrict to the hardest inputs, all goes well: we have that each $f \in \Func$ is an $n$-bit string with Hamming weight $1$, and ${\sf F}(f)$ is defined to be the unique index $i\in\Sigma=[n]$ such that $f(i) = 1$. However, if we relax $\Func$ to include all $n$-bit strings with Hamming weight at least $1$, then there are multiple correct indices $i$ such that $f(i) = 1$. Consequently, there is no longer a single correct value for ${\sf F}(f)$ for each $f \in \Func$. Further generalising $\Func$ to include all $n$-bit strings leads to cases where some inputs contain no indices mapping to $1$, making ${\sf F}(f)$ undefined for such inputs.

In search problems, the problem is therefore characterised by a \textit{relation} ${\cal R} \subset \Func \times \Sigma$, and the algorithm must output some $z \in \Sigma$ on input $f$ such that $(f, z) \in {\cal R}$. This formulation generalises the concept of computing a function ${\sf F}$, as we can define the relation ${\cal R}$ corresponding to ${\sf F}$ as the set $\{(f, {\sf F}(f)): f \in \Func\}$. We shall see in \thm{compressed} that the compressed oracle framework solves such search problems.

To remain closer to the notation used in \thm{mladv}, we still choose to frame search problems in terms of computing a function ${\sf F}$. To accommodate the fact that search problems can have multiple, or even no, correct outputs, we define that a quantum query algorithm ${\cal A}$ has successfully computed a function ${\sf F}$ on an input $f \in \Func$ if it outputs $z$ such that $z \in {\sf F}(f)$ (now allowed to be a \emph{set} of valid outputs). Consequently, if $\Sigma$ is the set of possible outputs, then each ${\sf F}(f)$ is a subset of $\Sigma$. To distinguish this from the earlier case where ${\sf F}(f)$ is a single value, we now write ${\sf F}: \Func \to 2^{\Sigma}$ for such search problems.

To reflect the modified definition of ``success'', we also update the projector $F_z$ for each $z \in \Sigma$ in \thm{madv} to:
$$
F_z = \sum\limits_{\substack{f \in \Func:\\{\sf F}(f) \ni z}} \proj{f}.
$$
We show that these modifications do not impact Item $3$ in \thm{madv}, thereby generalising \thm{madv} to search problems:
\begin{lemma}\label{lem:item3}
	Let $\Gamma$ be a multiplicative adversary matrix for a problem ${\sf F}: \Func \rightarrow 2^\Sigma$ and let $\lambda$ satisfy the constraints of \thm{madv}. Let $\Lambda_{{\sf bad}}$ be the projector onto the eigenspaces of $\Gamma$ corresponding to eigenvalues smaller than $\lambda$ and let $\eta\leq 1-\epsilon$ be a positive constant such that $\norm{F_z\Lambda_{{\sf bad}}}^2\leq \eta$ for every $z\in \Sigma$, where $F_z = \sum\limits_{\substack{f \in \Func: {\sf F}(f) \ni z}}\proj{f}$. 
	
	Then for any $T$-query quantum algorithm ${\cal A}$ that solves ${\sf F}$ on input $\ket{\delta}$ with success probability at least $1-\epsilon$, 
	$$W^T(\Gamma,{\cal A})\geq 1 + \left(\lambda - 1\right)\left(\sqrt{1-\epsilon} - \sqrt{\eta}\right)^2.$$
\end{lemma}
\begin{proof}
	Consider the final state $\ket{\psi_T({\cal A},\delta)}$ at the end of the computation. The output is correct if and only if $z \in {\sf F}(f)$, meaning we can define a ``success''  measurement on the input register ${\cal I}$ and the workspace register ${\cal W}_O$ containing the output $z\in \Sigma$:
	$$\Lambda_{{\sf succ}} := \sum_{z \in \Sigma} \proj{z}_{{\cal W}_O} \otimes F_z.$$
	
	Since the algorithm $\mathcal{A}$ solves $\mathcal{F}$ with success probability at least $1 - \epsilon$ on the input $\ket{\delta}$, we know that
	\begin{equation}\label{eq:succ}
		\norm{\Lambda_{\sf{succ}} \ket{\psi_T(\mathcal{A}, \delta)}} \geq \sqrt{1 - \epsilon}.
	\end{equation}
	
	As in the original proof of Item $3$ in~\cite{vspalek2008multiplicative}, we define $\Lambda_{\sf{good}} := I - \Lambda_{\sf{bad}}$ as the projector onto the orthogonal complement of the bad subspace, which we call the good subspace. Using these projectors, we decompose $\ket{\psi_T(\mathcal{A}, \delta)}$ as follows:
	\begin{equation}\label{eq:good-bad}
		\ket{\psi_T(\mathcal{A}, \delta)} = \sqrt{1 - \beta} \ket{\Psi_{\sf{bad}}} + \sqrt{\beta} \ket{\Psi_{\sf{good}}},
	\end{equation}
	where
	$$
	\ket{\Psi_{\sf{bad}}} = \frac{\Lambda_{\sf{bad}} \ket{\psi_T(\mathcal{A}, \delta)}}{\|\Lambda_{\sf{bad}} \ket{\psi_T(\mathcal{A}, \delta)}\|}, ~~ \ket{\Psi_{\sf{good}}} = \frac{\Lambda_{\sf{good}} \ket{\psi_T(\mathcal{A}, \delta)}}{\|\Lambda_{\sf{good}} \ket{\psi_T(\mathcal{A}, \delta)}\|}, \text{ and } ~~ \beta = \norm{\Lambda_{\sf{good}} \ket{\psi_T(\mathcal{A}, \delta)}}^2.
	$$
	
	We proceed by separately bounding the contributions of the ``good'' and ``bad'' components to $\norm{\Lambda_{\sf{succ}} \ket{\psi_T(\mathcal{A}, \delta)}}$. For the ``good'' component, we can use the trivial bound, namely $\norm{\Lambda_{\sf{succ}} \ket{\Psi_{\sf{good}}}} \leq 1$. For the ``bad'' component, we bound it by
	$$
	\norm{\Lambda_{\sf{succ}} \ket{\Psi_{\sf{bad}}}} \leq \max_{z \in \Sigma} \norm{F_z \Lambda_{\sf{bad}}} \leq \sqrt{\eta}.
	$$
	Combining this with \eq{good-bad} and \eq{succ}, we find that
	$$\sqrt{1-\epsilon} \leq \norm{\Lambda_{\sf{succ}} \ket{\psi_T(\mathcal{A}, \delta)}} \leq \sqrt{1-\beta}\norm{\Lambda_{\sf{succ}}\ket{\Psi_{\sf{bad}}}} + \sqrt{\beta}\norm{\Lambda_{\sf{succ}}\ket{\Psi_{\sf{good}}}} \leq \sqrt{\eta} + \sqrt{\beta},$$
	which we can rearrange to obtain $\beta \geq \left(\sqrt{1 - \epsilon} - \sqrt{\eta}\right)^2$. 
	
	Having found a lower bound on $\beta$, we can now apply the same decomposition from \eq{good-bad} to our progress measure to conclude the lemma:
	\begin{align*}
		W^T(\Gamma, \mathcal{A}) &= \Tr(\Gamma \rho_{\mathcal{A}}^T(\mathcal{A}, \delta)) \geq \Tr(\lambda \Lambda_{\sf{good}} \rho_{\mathcal{A}}^T(\mathcal{A}, \delta)) + \Tr(\Lambda_{\sf{bad}} \rho_{\mathcal{A}}^T(\mathcal{A}, \delta)) \\
		& \geq \lambda \beta + (1 - \beta) \geq 1 + (\lambda - 1) \left(\sqrt{1 - \epsilon} - \sqrt{\eta}\right)^2. \qedhere
	\end{align*}
\end{proof}

\subsection{The compressed oracle technique}\label{sec:comp-oracle}

In the compressed oracle technique~\cite{zhandry2019record}, Zhandry observes that in query problems where the algorithm interacts with a quantum random oracle, it is equivalent (by applying a purification) to assume that the algorithm is run on a uniform superposition over all possible functions from the set $X$ to the set $Y$. In this picture, a quantum adversary interacting with the quantum random oracle towards some nefarious end is analogous to a quantum algorithm run on input distribution $\delta$, which is initialised to the uniform distribution over all functions from $X$ to $Y$:
\begin{equation}\label{eq:deltacomp}
	\ket{{\sf Uniform}}_{\mathcal{I}} \coloneqq \frac{1}{\sqrt{M^{N}}} \sum_{f \in Y^X} \ket{f}_{\mathcal{I}}.
\end{equation}

We refrain from discussing the compressed oracle in depth here. For more details, see~\cite{zhandry2019record, czajkowski2019quantum, hamoudi2020quantum, chung2021compressed}. Instead, we summarise the necessary parts needed to show how the compressed oracle technique can be used to derive quantum query lower bounds whenever $\Func = Y^X$. The input register ${\cal I}$ holding any computational basis state $\ket{f}_{\cal I}$, where $f \in Y^X$, can be viewed as a tensor product of the different function values for $f$ for different values of $x \in X$:
$$ \ket{f}_{\cal I} = \bigotimes_{x \in X}\ket{f(x)}_{{\cal I}_x}. $$

This can be interpreted as a look-up table that fully describes the action of $f$. We can also consider a Fourier basis (see \defin{fourier}) for this register that represents a function in $Y^X$. Let $\{\ket{\hat{f}}\}_{f \in Y^X}$ be the \emph{Fourier basis} of ${\cal I}\equiv {\cal Y}^{\otimes N}$, where each $\ket{\hat{f}}$ is defined as
$$ \ket{\hat{f}}_{{\cal I}} \coloneqq \bigotimes_{x \in X}{\sf QFT}_M \ket{f(x)}_{{\cal I}_x} =  \bigotimes_{x \in X} \ket{\widehat{f(x)}}_{{\cal I}_x}. $$
From this look-up table perspective, this means that we change the basis of all our entries in the look-up table. The key insight that Zhandry makes is that if we view both the input register $\mathcal{I}$ in this Fourier basis, as well as the ${\cal Y}$ register, then a query (as in \defin{query}) acts on a basis state $\ket{x}_{\mathcal{X}}\ket{\hat{y}}_{\mathcal{Y}}\ket{\hat{f}}_{\mathcal{I}}$ as follows:
\begin{align}\label{eq:fourier-query}
	{\cal O}\left(\ket{x}_{\mathcal{X}}\ket{\hat{y}}_{\mathcal{Y}}\ket{\hat{f}}_{\mathcal{I}}\right) =  \ket{x}_{\mathcal{X}}\ket{\hat{y}}_{\mathcal{Y}}\ket{\widehat{f - {y}\cdot\delta_{x}}}_{\mathcal{I}}.
\end{align}
Here, $\delta_x$ denotes the point function satisfying $\delta_x(x) = 1$ and $\delta_x(x') = 0$ for all $x' \neq x$, so $f-y\cdot \delta_x$ is the function that agrees with $f$ on all values except possibly $x$, where it takes value $f(x)-y$. This change of perspective is quite peculiar: where in a regular query (as in \defin{query}) the information stored in the ${\cal I}_x$ register is ``copied'' into the ${\cal Y}$ register, this interaction is mirrored when viewing the ${\cal I}_x$ register in the Fourier basis. Another added benefit of this basis change is that the initial state $\ket{{\sf Uniform}}$ simplifies to
\begin{align}\label{eq:initial}
	{\sf QFT}_M^{\otimes N} \ket{{\sf Uniform}}_{\mathcal{I}} = \bigotimes_{x \in X}\ket{\hat{0}}_{\mathcal{I}_x}.
\end{align}

The action of the oracle in \eq{fourier-query}, combined with \eq{initial}, implies the following consequence, which is the cornerstone of the compressed oracle technique:
\begin{fact}\label{fct:size}
	For any $T$-query quantum algorithm ${\cal A}$ and for any $t \in [T]_0$, we have that $\rho^t_{\mathcal{I}}({\cal A},{{\sf Uniform}})$ is supported on vectors in the Fourier basis of the form $\ket{\hat{f}}$ where 
	$$ {f} = {y_1}\cdot\delta_{x_1} + \cdots + {y_s}\cdot\delta_{x_s}, $$
	for some $x_1,\dots,x_s \in X$, $y_1,\dots y_s \in Y$, and $s \in [t]_0$.  
\end{fact}
\noindent In \lem{space}, we will establish a stronger relationship that directly implies \fct{size}.

We can construct an isometry $\mathsf{Comp}_x: \mathbb{C}[Y] \rightarrow \mathbb{C}[Y \cup \{\bot\}]$, for every $x \in X$, that maps the $\mathcal{I}_x$ register to $\ket{\bot}$ if and only if this register contains $\ket{\hat{0}}$, which represents the algorithm knowing nothing about the value stored in register ${\cal I}_x$:
\begin{align*}
	\mathsf{Comp}_x = \ket{\bot}\bra{\hat{0}} + \sum_{z \in Y \setminus \{0\}}\proj{\hat{z}}.
\end{align*}
By doing this for every $x \in X$ we obtain the isometry
$$ \mathsf{Comp} = \bigotimes_{x \in X} \mathsf{Comp}_x. $$
This isometry $\mathsf{Comp}$ \textit{compresses} the information of each of the basis vectors $\ket{\hat{f}}$, for $f=y_1\delta_{x_1}+\dots+ y_s\delta_{x_s}$, in the support of $\rho^t_{\mathcal{I}}({\cal A},{{\sf Uniform}})$, since $\mathsf{Comp}\ket{\hat{f}} \in \mathbb{C}[\left(Y \cup \{\bot\}\right)^X]$ has $\ket{\bot}$ everywhere except for those $s\leq t$ registers indexed by $x_1,\dots,x_s$. Let us extend ${\sf QFT}_M$ to $\mathbb{C}[\left(Y\cup\{\bot\}\right)^X]$ by defining ${\sf QFT}_M\ket{\bot}=\ket{\bot}$. We can view
$$\ket{D}={\sf QFT}_M{\sf Comp}\ket{\hat{f}}\in\mathbb{C}[\left(Y \cup \{\bot\}\right)^X]$$
as a \emph{database}, where we have applied ${\sf QFT}_M$ to bring the databases back to the computational basis. We say that $D$ has size $s$ if $\abs{\{x \in X: D(x) \neq \bot\}} = s$, which we denote by $\abs{D} = s$, and remark that by ${\sf QFT}_M\ket{\bot}=\ket{\bot}$, this basis conversion leaves the size of the database unaffected. We write
\begin{align}\label{eq:size}
	&{\cal D}_{s} \coloneqq \{D \in \left(Y \cup \{\bot\}\right)^X: \abs{D} = s\}, &{\cal D}_{\leq s} \coloneqq \{D \in \left(Y \cup \{\bot\}\right)^X: \abs{D} \leq s\},
\end{align}
for the sets of all databases of size $s$ and at most $s$, respectively. We let ${\cal D}=(Y\cup\{\bot\})^X$ denote the set of all databases of any size.
 In this work, we use set notation when working with databases: 
\begin{itemize}
	\item For any $x \in X, y \in Y$ and $D \in \left(Y \cup \{\bot\}\right)^X$ such that $D(x) = \bot$, we can add a new entry $(x,y)$ to $D$, to obtain $D' = D \cup (x,y)$. This means that the resulting database $D'$ satisfies $D(x') = D'(x')$ for every $x' \in X \setminus \{x\}$ and $D'(x) = y$.
	\item  For any $x \in X, y \in Y$ and $D \in \left(Y \cup \{\bot\}\right)^X$ such that $D(x) = y$, we can delete the entry $(x,y)$ from $D$, to obtain $D' = D \setminus (x,y)$. This means that the resulting database $D'$ satisfies $D(x') = D'(x')$ for every $x' \in X \setminus \{x\}$ and $D'(x) = \bot$.
\end{itemize}

The compressed oracle gets its name from the fact that each database $D$ of size $s$ can be efficiently represented by the list of pairs $(x_1,D(x_1)),\dots,(x_s,D(x_s))$, which is bounded in size due to \fct{size}. Hence, the oracle operation ${\cal O}_{x,y}$ can be efficiently computed by a quantum algorithm that lazy samples from the uniform distribution, and this circuit (see~\cite{czajkowski2019quantum} for its explicit construction) is referred to as the compressed (Fourier) oracle:
\begin{align}\label{eq:cfo}
	\mathsf{cO}_{x,y} = \mathsf{Comp} \circ {\cal O}_{x,y} \circ \mathsf{Comp}^{\dagger}.
\end{align}

This framework has many applications in cryptography~\cite{czajkowski2019quantum,liu2019revisiting,grilo2021tight,don2022online} by being able to analyse the interaction of an adversary with a random oracle, which as we have seen is equivalent to where the input register $\mathcal{I}$ is initialised to the uniform superposition over all functions (see \eq{deltacomp}). In~\cite{czajkowski2019quantum, hamoudi2020quantum}, it was shown that this can be generalised to uniform superpositions over distributions where there is no correlation between the values in the registers ${\cal I}_x$ and ${\cal I}_{x'}$ for distinct $x,x' \in X$. In this work, we focus only on the application of the compressed oracle technique to quantum query lower bounds. A rigorous framework of this application has been given in~\cite{chung2021compressed}, where the main ingredient of this lower bound (see \thm{compressed} for the full statement) is of the following form:
$$ \max\limits_{x \in X,y \in Y}\norm{\mathsf{P}_{{\cal D}_{{\cal P}}}\mathsf{cO}_{x,y}\left(I - \mathsf{P}_{{\cal D}_{{\cal P}}}\right)}. $$
Here, the property ${\cal P} \subseteq (X \times Y)^k$ defines a set of tuples of size $k$ over $X \times Y$. Each tuple $p \in {\cal P}$ is an element of $(X \times Y)^k$ and represents a list of input-output pairs $((x_1, y_1), \dots, (x_k, y_k))$. A property ${\cal P}$ induces a relation ${\cal R}$ on the input $f \in \Func$, as discussed in \sec{search}, by saying that for $p = (x_1, y_1), \dots, (x_k, y_k) \in {\cal P}$, we have $(f,p) \in {\cal R}$ if and only if the input-output pairs in $p$ are consistent with the input $f$. As an example, consider the collision problem, where for any input $f \in Y^X$, the goal is to output a pair $(x_1, y)$, $(x_2, y)$ such that $f(x_1) = f(x_2) = y$, referred to as a \textit{collision}. The corresponding property ${\cal P}$ in this case would be
$$
{\cal P} = \{((x_1, y_1), (x_2, y_2)) \in \left(X \times Y\right)^2:y_1=y_2\}.
$$

A property ${\cal P}$ also induces a subset ${\cal D}_{\cal P} \subseteq {\cal D}$, where $D \in {\cal D}_{\cal P}$ if and only if it is consistent with one of the tuples in ${\cal P}$, meaning that there exists a $k \in [N]$ and $p = ((x_1,y_1),\dots,(x_k,y_k)) \in {\cal P}$ such that $D(x_1) = y_1,\dots,D(x_k) = y_k$. For any subset $A \subseteq {\cal D}$, we denote the projection onto this subset as
\begin{align}\label{eq:proj}
	\mathsf{P}_{A} = \sum_{D \in A} \proj{D}.
\end{align}
Since these projectors project onto computational basis states, we have the added benefit that they commute for distinct choices of $A$.

\begin{theorem}[{\cite{chung2021compressed}}]\label{thm:compressed}
	Fix a finite set $X$ of size $N$ and let $Y = [M-1]_0$. Let ${\cal P} \subseteq (X \times Y)^k$ be a property for some $k \in [M-1]$ and consider a quantum algorithm ${\cal A}$ that outputs $(x_1,y_1),\dots,(x_k,y_k)$. Let $p$ be the probability that both $((x_1,y_1),\dots,(x_k,y_k)) \in {\cal P}$ and $y_i = f(x_i)$ for every $i \in [k]$ when ${\cal A}$ has interacted with a random oracle, initialised with a uniformly random function $f$ in $Y^X$. Then:
	\begin{align*}
		\sqrt{p} \leq \sum_{t = 1}^T \max\limits_{x \in X,y \in Y}\norm{\mathsf{P}_{{\cal D}_{\leq t} \cap {\cal D}_{{\cal P}}}\mathsf{cO}_{x,y}\mathsf{P}_{{\cal D}_{\leq t-1} \setminus {\cal D}_{{\cal P}}}} + \sqrt{\frac{k}{M}}.
	\end{align*}    
\end{theorem}

\begin{remark}
	The framework in~\cite{chung2021compressed} allows for an adversary that makes both sequential as well as parallel queries, whereas we restrict to only the sequential query version of their result. Moreover, they also allow for a series of properties ${\cal P}_0,\dots,{\cal P}_T$ instead of a single property ${\cal P}$, where they bound
	$$\norm{\mathsf{P}_{{\cal D}_{\leq t} \cap {\cal D}_{{\cal P}_t}}\mathsf{cO}_{x,y}\mathsf{P}_{{\cal D}_{\leq t-1} \setminus {\cal D}_{{\cal P}_{t-1}}}}.$$
Since the latter generalisation has thus far not been used for any application in the sequential query model, we consider the simplified lower bound as described and applied in~\cite{zhandry2019record, liu2019finding,hamoudi2020quantum}.
\end{remark}

The form of \thm{compressed} is restricted compared to that of \thm{madv}. We saw that we cannot run ${\cal A}$ on any input distribution, but only on ${\sf Uniform}$, since \thm{compressed} requires the register ${\cal I}$ to be initialised with a uniformly random function $f$ in $Y^X$. Since ${\cal A}$ has to output $((x_1,y_1),\dots,(x_1,y_k)) \in {\cal P} \subseteq \left(X \times Y\right)^{k}$, the technique always deals with search problems instead of decision problems. Despite this restriction, it does seem to come with a large advantage compared to the adversary methods. In practice, it appears to be much more straightforward, or at least more intuitive, to come up with a good bound on $\|\mathsf{P}_{{\cal D}_{\leq t} \cap {\cal D}_{{\cal P}}}\mathsf{cO}_{x,y}\mathsf{P}_{{\cal D}_{\leq t-1} \setminus {\cal D}_{{\cal P}}}\|$ than it is to derive a good multiplicative adversary matrix $\Gamma$ and accompanying constants $\lambda,\eta$, and bound its progress $\|{\cal O}_{x,y}^{\dagger}\Gamma^{1/2} {\cal O}_{x,y}\Gamma^{-1/2}\|$. Furthermore, like the multiplicative adversary method, it also works well when one considers exponentially small success probabilities, whereas the negative-weights adversary method fails in this regime.

\subsection{Average-case query complexity}

\thm{compressed}, as stated, does not explicitly give a lower bound on $Q_{\epsilon}({\sf F})$, but it does imply one: Recall from \defin{complex} and \sec{search} that $Q_{\epsilon}({\sf F})$ captures the number of queries required for any quantum query algorithm ${\cal A}$ to successfully output $z \in {\sf F}(f)$ for \emph{any} input $f \in \Func$ with success probability at least $1-\epsilon$. By convexity, $Q_{\epsilon}({\sf F})$ is lower bounded by the number of queries required for {any} input distribution $\delta$, since:
$$
\Pr_{f \sim \delta}[{\cal A} \text{ outputs } z \in {\sf F}(f)] \geq \min_{f \in \Func} \Pr[{\cal A} \text{ outputs } z \in {\sf F}(f)].
$$

However, $Q_{\epsilon}({\sf F})$ is not an interesting metric in the case where $\min_{f \in \Func} \Pr[{\cal A} \text{ outputs } z \in {\sf F}(f)]$ could be $0$, i.e.~when there exists an input $f \in \Func$ when the algorithm \textit{can't} successfully output $z \in {\sf F}(f)$ for any input $f \in \Func$. This can occur in \thm{compressed}, as the input distribution $\delta$ is ${\sf Uniform}$. For instance, recall the collision problem. Some inputs $f \in Y^X$ may contain no collisions, making it impossible for the quantum algorithm to output $z \in {\sf F}(f)$. 

Additionally, even if the worst-case input admits a non-zero probability of success, it can often be more meaningful to show that the problem is hard \textit{on average} rather than merely demonstrating the existence of an input where the problem is hard. This is particularly relevant in the context of cryptography, where it is more desirable to know that a randomly chosen security key yields a secure construction than to prove that there exists a single specific key ensuring security. Therefore, in the remainder of this work, we focus on deriving a lower bound for the \textit{average-case} complexity $Q_{\epsilon}({\sf F})$  rather than the \textit{worst-case} complexity $Q_{\epsilon}({\sf F})$:
\begin{definition}[$\epsilon$-error Average-Case Quantum Query Complexity]\label{def:complex-gen}
	Fix ${\sf F}:{\sf Func}\rightarrow 2^\Sigma$. Then the \emph{$\epsilon$-error average-case quantum query complexity} of ${\sf F}$ and input distribution $\delta$ on $\Func$, denoted by $Q_{\epsilon}^{\delta}({\sf F})$, is the minimum number of queries needed by any quantum query algorithm ${\cal A}$ such that 
	$$\Pr_{f \sim \delta}[{\cal A}\text{ outputs }z \in {\sf F}(f)] \geq 1 - \epsilon.$$
\end{definition}

\noindent We can now use \thm{compressed} to lower bound $Q_{\epsilon}^{\sf Uniform}({\sf F})$:
\begin{corollary}\label{cor:compressed}
	Fix a finite set $X$ of size $N$ and let $Y = [M-1]_0$. Let ${\cal P} \subseteq (X \times Y)^k$ be a property for some $k \in [M-1]$. Then for any $\epsilon\in \left(0,1-k/M\right)$ and any problem ${\sf F}: Y^X \rightarrow 2^{\cal P}$, the $\epsilon$-error average-case quantum query complexity $Q_{\epsilon}^{\sf Uniform}({\sf F})$ is lower bounded by the smallest $T$ satisfying
	\begin{equation*}
		\sqrt{1-\epsilon} - \sqrt{\frac{k}{M}} \leq \sum_{t = 1}^T \max\limits_{x \in X,y \in Y}\norm{\mathsf{P}_{{\cal D}_{\leq t} \cap {\cal D}_{{\cal P}}}\mathsf{cO}_{x,y}\mathsf{P}_{{\cal D}_{\leq t-1} \setminus {\cal D}_{{\cal P}}}}.
	\end{equation*}
\end{corollary}
\cor{compressed} is slightly less conveniently phrased compared to \cor{madv} due to its dependence on $t$ in the term 
$$ 
\norm{\mathsf{P}_{{\cal D}_{\leq t} \cap {\cal D}_{{\cal P}}} \mathsf{cO}_{x,y} \mathsf{P}_{{\cal D}_{\leq t-1} \setminus {\cal D}_{{\cal P}}}}. 
$$ 
As an example of how to determine the ``smallest $T$'' in \cor{compressed}, we consider the collision problem. In~\cite{zhandry2019record}, it is shown that for the collision property 
$$
{\cal P} = \{((x_1, y_1), (x_2, y_2)) \in \left(X \times Y\right)^2:y_1=y_2\},
$$
we can bound
$$
\max\limits_{x \in X, y \in Y} \norm{\mathsf{P}_{{\cal D}_{\leq t} \cap {\cal D}_{{\cal P}}} \mathsf{cO}_{x,y} \mathsf{P}_{{\cal D}_{\leq t-1} \setminus {\cal D}_{{\cal P}}}} \leq \sqrt{\frac{t-1}{M}}.
$$
Hence, $Q_{\epsilon}^{\sf Uniform}({\sf F})$ is lower bounded by the smallest $T$ satisfying:
$$
\sqrt{1-\epsilon} - \sqrt{\frac{2}{M}} \leq \sum_{t=1}^T \sqrt{\frac{t-1}{M}} \leq \frac{T^{3/2}}{\sqrt{M}},
$$
which can be rearranged to yield
$$
T \geq \left(\sqrt{1-\epsilon} - \sqrt{\frac{2}{M}}\right)^{2/3} M^{1/3}.
$$

\section{Multiplicative ladder adversary method}\label{sec:mladv}

Here, we propose a simplified version of the multiplicative adversary method, that we name the \textit{multiplicative ladder adversary} (MLA) method, which we later prove has the compressed oracle technique as a special case (see \sec{reduction}) as well as the polynomial method (see \sec{reduction-poly}). The MLA method is weaker than the multiplicative adversary method as it only considers a subset of all possible multiplicative adversary matrices $\Gamma$, which we refer to as MLA matrices, but despite this restriction, it still exhibits a strong direct product theorem, as will be shown in \sec{dpt}.

\subsection{Making the adversary matrix time-dependent}

Before we define these MLA matrices in \defin{mladv}, we first provide some motivation behind their definition. In \sec{comp-oracle}, we saw that the compressed oracle seems to make more explicit use of the number of queries to compute the incremental progress by decomposing the set of all possible databases ${\cal D} = \mathop{\bigsqcup}_{t=0}^N {\cal D}_t$ based on their sizes and integrating these into the projection $\mathsf{P}_{{\cal D}_{{\cal P}}}$. We generalise this notion by introducing the following construction, that captures the subspace of $\mathbb{C}[Y^X]$ that is reachable from $\ket{\delta}$ after a fixed number of queries.

First, we define a few components necessary for our construction. Let $\delta$ be an initial distribution on $\Func \subseteq Y^X$. For any $t \in [N]$ and any choice of $x_1, \dots, x_t \in X$ and $y_1, \dots, y_t \in Y$, define
\begin{equation}\label{eq:vstate}
	\ket{v_{x_1,\dots,x_t}^{y_1,\dots,y_t}} :=   \frac{1}{\sqrt{\alpha_{x_1,\dots,x_t}^{y_1,\dots,y_t}}}\sum_{\substack{f \in \Func:\\ \forall i \in [t], f(x_i) = y_i}}\sqrt{\delta(f)}\ket{f},
\end{equation}
where $\alpha_{x_1,\dots,x_t}^{y_1,\dots,y_t}$ is the normalisation factor, defined as
\begin{equation}\label{eq:vstate-coef}
	\alpha_{x_1,\dots,x_t}^{y_1,\dots,y_t} := \sum_{\substack{f \in \Func:\\ \forall i \in [t],  f(x_i) = y_i}} \delta(f).
\end{equation}

\begin{lemma}\label{lem:space}
	Define the sequence of subspaces ${\sf Space}_t(\delta)$ as follows:
	\begin{itemize}
		\item For $t = 0$, let ${\sf Space}_0(\delta) = \mathrm{span}\{\ket{\delta}\}$, where $\ket{\delta} = \sum_{f \in \Func} \sqrt{\delta(f)} \ket{f}$ is the initial state of the input register ${\cal I}$.
		\item For $t \in [N]$, set
		\[
		{\sf Space}_t(\delta) \coloneqq \mathrm{span}\left\{ \ket{v_{x_1, \dots, x_t}^{y_1, \dots, y_t}} : (x_i, y_i) \in X \times Y \text{ for } i = 1, \dots, t \right\}.
		\]
		\item For $t > N$, define ${\sf Space}_t(\delta) = {\sf Space}_N(\delta)$.
	\end{itemize}
	
	Then each space ${\sf Space}_t(\delta)$ represents the subspace of $\mathbb{C}[Y^X]$ that is reachable from $\ket{\delta}$ after $t$ queries:
	\begin{itemize}
		\item For every $t \in [N]_0$, there exists a $t$-query quantum algorithm ${\cal A}$, such that
		$${\sf Space}_{t}(\delta) \subseteq \mathrm{span}\left\{ \mathrm{supp}\left(\rho_{\mathcal I}^t({\cal A}, \delta)\right)\right\}.$$ 
		\item For every $t \in [N]_0$ and $t$-query quantum algorithm ${\cal A}$, we have
		$${\sf Space}_{t}(\delta) \supseteq \mathrm{span}\left\{ \mathrm{supp}\left(\rho_{\mathcal I}^t({\cal A}, \delta)\right)\right\}.$$ 
	\end{itemize}
\end{lemma}

Before we prove the lemma, we discuss some of its implications. First of all, we find that ${\sf Space}_{N}(\delta)=\mathrm{span}\{\ket{f}: f \in \mathrm{supp}(\delta)\}$. Moreover, in the special case where $\ket{\delta} = \ket{{\sf Uniform}}$, we have that
\begin{equation}\label{eq:space-t-uniform}
	{\sf Space}_t({{\sf Uniform}}) = \mathsf{Comp}^{\dagger}\left(\mathrm{span}\{\ket{D}: D \in {\cal D}_{\leq t}\}\right)\mathsf{Comp},
\end{equation} 
which recovers \fct{size}. 

We can combine $\Gamma$ with the projection $\Pi_{\leq t}$ that projects onto ${\sf Space}_t(\delta)$, to ensure that the progress keeps track of the number of queries done by the algorithm. The ``$\leq$'' in the subscript of each of the projectors $\Pi_{\leq t}$ is there to emphasise that $\Pi_{\leq t-1} \preceq \Pi_{\leq t}$. This is due to the fact that we can let $(x_t,y_t) = (x_{t-1},y_{t-1})$ in $\ket{v_{x_1, \dots, x_t}^{y_1, \dots, y_t}}$. For any $T$-quantum algorithm $\cal A$, initial distribution $\delta$, $t \in [T]_0$, and multiplicative adversary matrix $\Gamma$, we have
\begin{align}\label{eq:adv-time}
	W^t(\Gamma,{\cal A}) = \Tr\left[\Gamma \rho_{\cal I}^t({\cal A},\delta)\right] = \Tr\left[\Gamma\Pi_{\leq t}\rho_{\cal I}^t({\cal A},\delta)\right] = \Tr\left[\Gamma\rho_{\cal I}^t({\cal A},\delta)\Pi_{\leq t}\right].
\end{align}

\begin{proof}[Proof of \lem{space}]
	To prove the first inclusion in \lem{space}, we show that for any fixed $\bm{x} = (x_1,\dots,x_t) \in X^t$ and $\bm{y} = (y_1,\dots,y_t) \in Y^t$, we can construct a quantum algorithm $A$ such that $$\ket{v_{x_1,\dots,x_t}^{y_1,\dots,y_t}} \in \mathrm{supp}\left(\rho_{\mathcal I}^t({\cal A}, \delta)\right).$$
	
	Let ${\cal A}$ be the $t$-query algorithm that computes $\ket{f(x_1),\dots,f(x_t)}_{\cal W}$ in its working register using $t$ queries and $t+1$ unitaries. Additionally, its final unitary $U_t$ uncomputes the ${\cal Y}$ register. Then the final state of the algorithm $A$ is
	\begin{equation*}
		\ket{\psi_t({\cal A},\delta)} = \sum_{f \in \Func}\sqrt{\delta(f)}\ket{f(x_1),\dots,f(x_t)}_{\cal W}\ket{x_t}_{\cal X}\ket{0}_{\cal Y}\ket{f}_{\mathcal I}.
	\end{equation*}
	By tracing out all but the input register ${\cal I}$, we obtain (see \eq{vstate} and \eq{vstate-coef}):
	\begin{equation}\label{eq:state-contain}
		\rho_{\mathcal I}^t({\cal A}, \delta) = \sum_{\bm{y}\in Y^t}\frac{1}{\alpha_{\bm{x}}^{\bm{y}}}\proj{v_{\bm{x}}^{\bm{y}}}.
	\end{equation}
	
	For the second inclusion in \lem{space}, we show inductively that for every quantum algorithm ${\cal A}$ and $t \in [N]_0$, we have
	$$\ket{\psi_t({\cal A},\delta)} \in {\cal WXY} \otimes {\sf Space}_{t}(\delta).$$
	For $t=0$, we know for any quantum algorithm ${\cal A}$ that 
	$$\ket{\psi_0({\cal A},\delta)} = 
	U_0\ket{0}_{\cal WXY} \otimes \ket{\delta}_{\mathcal I} \in {\cal WXY} \otimes {\sf Space}_{t}(\delta),$$
	since $U_0$ acts non-trivially only on the ${\cal WXY}$ registers. Now suppose that \eq{state-contain} holds for some choice of $t \in [N-1]_0$ and quantum algorithm ${\cal A}$, meaning there exist complex coefficients $\beta_{w,x,y,\bm{x},\bm{y}}$ satisfying
	$$\ket{\psi_t({\cal A},\delta)} = \sum_{\substack{x \in X, y \in Y, w \in W,\\\bm{x}\in X^t, \bm{y} \in Y}}\beta_{w,x,y,\bm{x},\bm{y}}\ket{w,x,\hat{y}}_{\cal WXY}\ket{v_{\bm{x}}^{\bm{y}}}_{\cal I}.$$
	Here $\ket{v_{\bm{x}}^{\bm{y}}}_{\cal I} = \ket{\delta}$ if $t = 0$. For each $x \in X$, we can decompose the state $\ket{v_{\bm{x}}^{\bm{y}}}$ based on the value of $f(x)$ in the computational basis states $\ket{f}$ in $\ket{v_{\bm{x}}^{\bm{y}}}$:
	\begin{equation}\label{eq:decomp-x}
		\ket{\psi_t({\cal A},\delta)} = \sum_{\substack{x \in X, y \in Y, w \in W,\\\bm{x}\in X^t, \bm{y} \in Y}}\beta_{w,x,y,\bm{x},\bm{y}}\ket{w,x,\hat{y}}_{\cal WXY}\left(\frac{1}{\sqrt{\alpha_{\bm{x}}^{\bm{y}}}}
		\sum_{y_{t+1} \in Y} \sqrt{\alpha_{\bm{x},x}^{\bm{y},y_{t+1}}} \ket{v_{\bm{x}, x}^{\bm{y}, y_{t+1}}}\right),
	\end{equation}
	which is an element of ${\cal WXY} \otimes {\sf Space}_{t+1}(\delta)$. Here $\alpha_{\bm{x}}^{\bm{y}} = 1$ if $t = 0$. For each such state $\ket{v_{\bm{x}, x}^{\bm{y}, y_{t+1}}}$, we only pick up a global phase when applying a phase query:
	\begin{equation}\label{eq:decomp-x-phase}
		{\cal O}\ket{x,\hat{y}}_{\cal XY}\ket{v_{x_1,\dots,x_t,x}^{y_1,\dots,y_t,y_{t+1}}} = e^{\frac{2\pi\iota}{M}y \cdot y_{t+1}}\ket{x,\hat{y}}_{\cal XY}\ket{v_{x_1,\dots,x_t,x}^{y_1,\dots,y_t,y_{t+1}}}.
	\end{equation}
	Moreover, since the unitary $U_{t}$ acts non-trivially only on the ${\cal WXY}$ registers, we find that 
	\begin{equation*}
		\ket{\psi_{t+1}({\cal A},\delta)} = U_t{\cal O}\ket{\psi_{t}({\cal A},\delta)} \in {\cal WXY} \otimes {\sf Space}_{t+1}(\delta).\qedhere
	\end{equation*}
\end{proof}

\subsection{Mapping the progress onto a ladder}

The structure of the database projections $\mathsf{P}_{{\cal D}_{\leq t} \cap {\cal D}_{{\cal P}}}$ and $\mathsf{P}_{{\cal D}_{\leq t-1} \setminus {\cal D}_{{\cal P}}}$ in the compressed oracle technique (see \thm{compressed}) is, in practice, more convenient to work with than the more abstract projections $\Lambda_i$. This is because these projections are built from the database basis states, which are more intuitive and allow for easy tracking of their sizes with each query (see \fct{size}). 

We aim to establish a similar structure on the eigenspaces of $\Gamma$. These eigenspaces should resemble steps on a ladder, where each query moves the state up or down by at most one step. Additionally, these steps should be evenly spaced. To formalise this idea, we impose structural constraints on the spectral decomposition of $\Gamma$:
\begin{equation}\label{eq:decomp-gamma}
	\Gamma = \sum_{i=0}^\ell \lambda_i \Lambda_i.
\end{equation}
Here, $\ell + 1$ denotes the number of distinct eigenvalues of $\Gamma$, which are sorted in ascending order, and each $\Lambda_i$ is the projector onto the eigenspace associated with the eigenvalue $\lambda_i$.

\begin{definition}[Multiplicative Ladder Adversary Matrix]\label{def:mladv}
	Let $\Gamma = \sum_{i=0}^\ell \lambda_i \Lambda_i$ be a multiplicative adversary matrix. We say that $\Gamma$ is a multiplicative ladder adversary (MLA) matrix if the following conditions hold:
	\begin{itemize}
		\item The eigenvalues of $\Gamma$ satisfy $\lambda_i = \kappa^i$ for some $\kappa > 1$, so that
		$$
		\Gamma = \sum_{i=0}^\ell \kappa^i \Lambda_i.
		$$
		\item For every $t \in [N]_0$, $\Gamma$ commutes with $\Pi_{\leq t}$.
		\item For all $x \in \mathcal{X}$, $y \in \mathcal{Y}$, and $i, i' \in [\ell]_0$, the projections onto the eigenspaces satisfy
		\begin{equation}\label{eq:ladder}
			\norm{\Lambda_{i'} \mathcal{O}_{x,y} \Lambda_i} = 0, \quad \text{if } \abs{i' - i} > 1.
		\end{equation}
	\end{itemize}
\end{definition}

The condition expressed in \eq{ladder} ensures that each query can move the state up or down by at most a single eigenspace. Meanwhile, the construction $\Gamma = \sum_{i=0}^\ell \kappa^i \Lambda_i$ ensures that the multiplicative progress between successive eigenspaces is constant, specifically a factor of $\kappa$.

This new definition allows us to prove an MLA-version of \thm{madv}. This result is strictly weaker, as it only considers a subset of all possible multiplicative adversary matrices, but it greatly simplifies the upper bound on the progress achievable in a single query (Item 2).
\begin{theorem}\label{thm:mladv}
	Fix a problem ${\sf F}: \Func \rightarrow 2^{\Sigma}$, an input distribution $\delta$ on $\Func$, a constant $\kappa > 1$, and an MLA matrix $\Gamma = \sum_{i = 0}^{\ell} \kappa^{i} \Lambda_i$ with $1$-eigenstate $\ket{\delta}$ (see \defin{mladv}). Let $\lambda$ be a real number with $1<\lambda\leq \kappa^{\ell}$. Let $\Lambda_{{\sf bad}}$ be the projector onto the eigenspaces of $\Gamma$ corresponding to eigenvalues smaller than $\lambda$ and let $\eta\leq 1-\epsilon$ be a positive constant such that $\norm{F_z\Lambda_{{\sf bad}}}^2\leq \eta$ for every $z\in \Sigma$, where $F_z = \sum\limits_{\substack{f \in \Func:\\{\sf F}(f) \ni z}}\proj{f}$. Then:
	\begin{enumerate}
		\item For any quantum algorithm ${\cal A}$, $W^0(\Gamma,{\cal A})=1$.
		\item For any $T$-query quantum algorithm ${\cal A}$, and $t\in[T-1]_0$,
		$$\displaystyle\frac{W^{t+1}(\Gamma,{\cal A})}{W^t(\Gamma,{\cal A})}\leq\max\limits_{\substack{i \in [\ell-1]_0, \\x \in X,y \in Y}}\left(1 + \max\limits_{\substack{i \in [\ell-1]_0, \\x \in X,y \in Y}}\frac{\kappa-1}{\sqrt{\kappa}}\norm{\Lambda_{i+1}\Pi_{\leq t+1}{\cal O}_{x,y}\Pi_{\leq t}\Lambda_{i}}\right)^2$$ 
		\item For any $T$-query quantum algorithm ${\cal A}$ that solves ${\sf F}$ on input $\ket{\delta}$ with success probability at least $1-\epsilon$, $W^T(\Gamma,{\cal A})\geq 1+(\lambda-1)(\sqrt{1-\epsilon}-\sqrt{\eta})^2$.
	\end{enumerate}
\end{theorem}

Note that the upper bound on $W^{t+1}/W^t$, the progress made in one step now depends on $t$. This is necessary to capture the power of the compressed oracle method, where, for example, the probability of (i.e.~amplitude on) finding a collision in a single query is greater the more queried values you have stored in memory.

\begin{corollary}\label{cor:mladv}
	For any $\eta$ that satisfies the constraints of \thm{mladv}, any $\epsilon\in (0,1-\eta)$, problem ${\sf F}:\Func\rightarrow2^\Sigma$, and input distribution $\delta$ on $\Func$, the $\epsilon$-error average-case quantum query complexity $Q_{\epsilon}^{\delta}({\sf F})$ is lower bounded by the smallest $T$ such that
	\begin{align*}
		1 \leq \min_{\Gamma, \lambda}\bigg(-(\lambda-1)(\sqrt{1-\epsilon}-\sqrt{\eta})^2 
		+ \prod_{t=1}^{T}\bigg(1 + \max\limits_{\substack{i \in [\ell-1]_0, \\x \in X,y \in Y}}
		\frac{\kappa-1}{\sqrt{\kappa}}\norm{\Lambda_{i+1}\Pi_{\leq t}{\cal O}_{x,y}\Pi_{\leq t-1}\Lambda_{i}}\bigg)^2\bigg),
	\end{align*}
\end{corollary}

\begin{proof}[Proof of \thm{mladv}]
	
	Item 1 and 3 follow from \thm{madv} and \lem{item3}, so we focus on proving Item 2. We fix any $T$-query algorithm ${\cal A}$ and begin by following similar steps as in~\cite{vspalek2008multiplicative} to upper bound the ratio
	$$\frac{W^{t+1}(\Gamma,{\cal A})}{W^t(\Gamma,{\cal A})}.$$
	
	Observe that the query operator ${\cal O}$ cannot directly be inserted into the progress measure since it acts on all registers $\cal{XYI}$, whereas each MLA matrix is only defined on $\cal{I}$. Thus, we lift $\Gamma$ to this larger space by constructing $\Upsilon = I_{\cal WXY} \otimes \Gamma$, which immediately yields
	$$W^t(\Gamma,{\cal A}) = \Tr\left[\Upsilon\proj{\psi_t({\cal A},\delta)}\right].$$ Because ${\cal A}$ and $\ket{\delta}$ are fixed, we simplify the notation in the rest of the proof and omit $({\cal A},\delta)$. The addition of $\Upsilon$ results in
	\begin{align*}
		\frac{W^{t+1}\left(\Gamma,{\cal A}\right)}{W^{t}\left(\Gamma,{\cal A}\right)} &= \frac{\Tr\left[\Upsilon\proj{\psi_{t+1}}\right]}{\Tr\left[\Upsilon\proj{\psi_t}\right]} = \frac{\Tr\left[\Upsilon U_{t+1}{\cal O}\proj{\psi_t}{\cal O}^{\dagger}U_{t+1}^{\dagger}\right]}{\Tr\left[\Upsilon\proj{\psi_t}\right]}\\
		& = \frac{\Tr\left[{\cal O}^{\dagger}U_{t+1}^{\dagger}\Upsilon U_{t+1}{\cal O}\proj{\psi_t}\right]}{\Tr\left[\Upsilon\proj{\psi_t}\right]},
	\end{align*}
	where in the final equality we have used the cyclic property of the trace. Since the unitary $U_{t+1}$ acts as the identity on register $\cal{I}$, we obtain that $U_{t+1}^{\dagger}\Upsilon U_{t+1} = \Upsilon$. This allows us to simplify 
	\begin{align*}
		\frac{\Tr\left[{\cal O}^{\dagger}U_{t+1}^{\dagger}\Upsilon U_{t+1}{\cal O}\proj{\psi_t}\right]}{\Tr\left[\Upsilon\proj{\psi_t}\right]} = \frac{\Tr\left[{\cal {\cal O}}^{\dagger}\Upsilon {\cal O}\proj{\psi_t}\right]}{\Tr\left[\Upsilon\proj{\psi_t}\right]}.
	\end{align*}
	
	At this point, we deviate from~\cite{vspalek2008multiplicative} by making use of the projection $\Pi_{\leq t}$ onto ${\sf Space}_t(\delta)$ from \lem{space}. Our goal is to show that the following equation holds:
	\begin{align}\label{eq:pure}
		\frac{W^{t+1}\left(\Gamma,{\cal A}\right)}{W^{t}\left(\Gamma,{\cal A}\right)} \leq \max\limits_{x \in X,y \in Y}\norm{\Gamma^{1/2} {\cal O}_{x,y}\Pi_{\leq t}\Gamma^{-1/2}}^2.
	\end{align}
	Let $\ket{\tau} = \Upsilon^{1/2}\ket{\psi_t}$, meaning $\ket{\psi_t} = \Upsilon^{-1/2}\ket{\tau}$. Then 
	\begin{align*}
		&\frac{\Tr\left[{\cal O}^{\dagger}\Upsilon {\cal O}\proj{\psi_t}\right]}{\Tr\left[\Upsilon\proj{\psi_t}\right]} = \frac{\bra{\psi_t}{\cal O}^{\dagger}\Upsilon {\cal O} \ket{\psi_t}}{\bra{\psi_t}\Upsilon \ket{\psi_t}} = \frac{\bra{\psi_t}\left(I_{\cal WXY} \otimes \Pi_{\leq t}\right){\cal O}^{\dagger}\Upsilon {\cal O} \left(I_{\cal WXY} \otimes \Pi_{\leq t}\right)\ket{\psi_t}}{\bra{\psi_t}\Upsilon \ket{\psi_t}} \\
		&= \frac{\bra{\tau}\Upsilon^{-1/2}\left(I_{\cal WXY} \otimes \Pi_{\leq t}\right){\cal O}^{\dagger}\Upsilon {\cal O}\left(I_{\cal WXY} \otimes \Pi_{\leq t}\right)\Upsilon^{-1/2} \ket{\tau}}{\braket{\tau}{\tau}}\\
		&\leq \norm{\Upsilon^{1/2} {\cal O}\left(I_{\cal WXY} \otimes \Pi_{\leq t}\right)\Upsilon^{-1/2}}^2 = \max\limits_{x \in X,y \in Y}\norm{\Gamma^{1/2}{\cal O}_{x,y}\Pi_{\leq t}\Gamma^{-1/2}}^2,
	\end{align*}
	where we use the fact that $\Gamma$, and hence also $\Upsilon$, is Hermitian, as well as \eq{query}. Using the triangle inequality, we can further bound this expression (for any fixed $x,y$ and without the square) as
	\begin{equation}\label{eq:norm-triangle}
		\begin{split}
			\norm{\Gamma^{1/2}{\cal O}_{x,y}\Pi_{\leq t}\Gamma^{-1/2}} &\leq \frac{1}{\sqrt{\kappa}} + \norm{\left(\Gamma^{1/2}{\cal O}_{x,y}\Pi_{\leq t} - \frac{1}{\sqrt{\kappa}}{\cal O}_{x,y}\Pi_{\leq t}\Gamma^{1/2}\right)\Gamma^{-1/2}}.
		\end{split}
	\end{equation}
	
	Here is where we make the second deviation from~\cite{vspalek2008multiplicative}. Since the projections onto the eigenspaces of a Hermitian matrix form a resolution of the identity, we can write the matrix 
	$$\left(\Gamma^{1/2}{\cal O}_{x,y}\Pi_{\leq t} - \frac{1}{\sqrt{\kappa}}{\cal O}_{x,y}\Pi_{\leq t}\Gamma^{1/2}\right)\Gamma^{-1/2}$$
	as a block matrix with entries indexed by $i,i' \in [\ell]_0$, equal to
	\begin{equation}\label{eq:entries}
		\begin{split}
			&\Lambda_{i'}\left(\Gamma^{1/2}{\cal O}_{x,y}\Pi_{\leq t} - \frac{1}{\sqrt{\kappa}}{\cal O}_{x,y}\Pi_{\leq t}\Gamma^{1/2}\right)\Gamma^{-1/2}\Lambda_{i} \\
			&= \left(\sqrt{\kappa^{i'}}\Lambda_{i'}{\cal O}_{x,y}\Pi_{\leq t} - \frac{1}{\sqrt{\kappa}}\Lambda_{i'}{\cal O}_{x,y}\Pi_{\leq t}\Gamma^{1/2}\right)\Gamma^{-1/2}\Lambda_{i} \\
			&= \frac{\sqrt{\kappa^{i'}}}{\sqrt{\kappa^{i}}}\Lambda_{i'}{\cal O}_{x,y}\Pi_{\leq t}\Lambda_{i} - \frac{1}{\sqrt{\kappa^1}}\Lambda_{i'}{\cal O}_{x,y}\Lambda_{i} = \frac{\sqrt{\kappa^{i'-i+1}} - 1}{\sqrt{\kappa}}\Lambda_{i'}{\cal O}_{x,y}\Pi_{\leq t}\Lambda_{i},
		\end{split}
	\end{equation}
	where we used $\Gamma = \sum_{i=0}^\ell\kappa^i\Lambda_{i}$ (see \eq{decomp-gamma}) and consequently $\Gamma^{-1}=\sum_{i=0}^\ell \kappa^{-i}\Lambda_{i}$. As $\Gamma$ is an MLA matrix, all entries in this block matrix must be zero by \eq{ladder}, apart from the entries on diagonal, superdiagonal and subdiagonal. The entries on the superdiagonal however are also zero, which can be verified by substituting $i' = i-1$ in \eq{entries}. This enables us to bound the norm of this block matrix by the block matrix $M_{x,y}$, which will contain only zero blocks, except for the blocks on the diagonal and subdiagonal, which have respective entries $a$ and $b$ (that depend on $x$ and $y$) multiplied by identity matrices of the appropriate dimensions. We set these values to
	\begin{align*}
		&a \coloneqq \max\limits_{i \in [\ell-1]_0} \norm{\frac{\sqrt{\kappa^{i-i+1}} - 1}{\sqrt{\kappa}}\Lambda_{i}{\cal O}_{x,y}\Pi_{\leq t}\Lambda_{i}} \leq 1 - \frac{1}{\sqrt{\kappa}}, \\
		&b \coloneqq \max\limits_{i \in [\ell-1]_0} \norm{\frac{\sqrt{\kappa^{(i+1)-i+1}} - 1}{\sqrt{\kappa}}\Lambda_{i+1}{\cal O}_{x,y}\Pi_{\leq t}\Lambda_{i}} 
		= \max_{i\in[\ell-1]_0}\frac{\kappa - 1}{\sqrt{\kappa}}\norm{\Lambda_{i+1}{\cal O}_{x,y}\Pi_{\leq t}\Lambda_{i}}.
	\end{align*}
	Using this new matrix $M_{x,y}$ we can therefore bound
	\begin{align}\label{eq:norm-M}
		\max\limits_{\substack{x \in X,y \in Y}}\norm{\left(\Gamma^{1/2}{\cal O}_{x,y}\Pi_{\leq t} - \frac{1}{\sqrt{\kappa}}{\cal O}_{x,y}\Pi_{\leq t}\Gamma^{1/2}\right)\Gamma^{-1/2}} \leq \max\limits_{\substack{x \in X,y \in Y}}\norm{M_{x,y}}.
	\end{align}
	For a block matrix of the form
	\begin{equation*}
		M_{x,y} = \begin{bmatrix} 
			a & 0 & 0 &  & \\
			b & a & 0 & & \\
			0 &b & \ddots & \ddots &  \\
			&  & \ddots & \ddots & 0\\
			&  &  & b & a\\
		\end{bmatrix},
	\end{equation*}
	we upper bound its spectral norm by $a+b$, since $\norm{M_{x,y}} \leq \sqrt{\norm{M_{x,y}}_1\norm{M_{x,y}}_{\infty}}$. We can now nearly conclude the proof by combining \eq{pure} with \eq{norm-triangle} and \eq{norm-M}:
	\begin{align*}
		\frac{W^{t+1}(\Gamma,{\cal A})}{W^t(\Gamma,{\cal A})} \leq \left(\frac{1}{\sqrt{\kappa}} + \max\limits_{\substack{x \in X,y \in Y}}\norm{M_{x,y}}\right)^2 \leq \left(1 + \max\limits_{\substack{i \in [\ell-1]_0, \\x \in X,y \in Y}}\frac{\kappa-1}{\sqrt{\kappa}}\norm{\Lambda_{i+1}{\cal O}_{x,y}\Pi_{\leq t}\Lambda_{i}}\right)^2.
	\end{align*} 
	
	To conclude Item $2$, we still need to replace $\Lambda_{i+1}{\cal O}_{x,y}\Pi_{\leq t}\Lambda_{i}$ with $\Lambda_{i+1}\Pi_{\leq t+1}{\cal O}_{x,y}\Pi_{\leq t}\Lambda_{i}$. This follows directly from \eq{decomp-x} and \eq{decomp-x-phase}, which shows that for every one of the basis states $\ket{v_{x_1,\dots,x_t}^{y_1,\dots,y_t}}$ spanning ${\sf Space}_t(\delta)$ and for every $x \in X$ and $y \in Y$, we have:
	\begin{align*}
		{\cal O}_{x,y}\ket{v_{x_1, \dots, x_t}^{y_1, \dots, y_t}} &= \frac{1}{\sqrt{\alpha_{x_1, \dots, x_t}^{y_1, \dots, y_t}}}
		\sum_{y_{t+1} \in Y} \sqrt{\alpha_{x_1, \dots, x_t,x}^{y_1, \dots, y_t,y_{t+1}}} {\cal O}_{x,y} \ket{v_{x_1, \dots, x_t, x}^{y_1, \dots, y_t, y_{t+1}}}\\
		&= \sum_{y_{t+1} \in Y} \sqrt{\alpha_{x_1, \dots, x_t,x}^{y_1, \dots, y_t,y_{t+1}}} e^{\frac{2\pi\iota}{M}y \cdot y_{t+1}}\ket{v_{x_1,\dots,x_t,x}^{y_1,\dots,y_t,y_{t+1}}} \in{\sf Space}_{t+1}(\delta),
	\end{align*}
	meaning 
	\begin{equation}\label{eq:one-query}
		{\cal O}_{x,y}\Pi_{\leq t} = \Pi_{\leq t+1}{\cal O}_{x,y}\Pi_{\leq t}. \qedhere
	\end{equation}
\end{proof}

The machinery of MLA matrices is not necessary for the reduction in \sec{reduction}. For this reduction, we construct multiplicative matrices with $\ell = 1$, which automatically satisfy \eq{ladder}. However, a general $\ell$ is required if we aim to compute a function ${\sf F}$ on $\ell$ independent instances simultaneously, as discussed in \sec{dpt}. Furthermore, almost all multiplicative adversary matrices constructed so far to establish lower bounds (see~\cite{ambainis2006new, vspalek2008multiplicative,ambainis2011symmetry}) are, in fact, MLA matrices. This observation suggests that MLA matrices form a natural subset worthy of deeper analysis.

The following is a useful property that follows from \eq{one-query}, which we will employ in the subsequent sections:
\begin{fact}\label{fct:mono}
	$\norm{\Lambda_{i}\Pi_{\leq t}{\cal O}_{x,y}\Pi_{\leq t-1}\Lambda_{i-1}}$ is monotonically non-decreasing in $t \in [N]$ for all $i \in [\ell]$. 	
\end{fact}
\begin{proof}
	Let $\Pi_t := \Pi_{\leq t} - \Pi_{\leq t-1}$ be the projection onto ${\sf Space}_t(\delta) \cap {\sf Space}_{t-1}(\delta)^{\bot}$. Since unitaries preserve inner products, we have by \eq{one-query} that
	\begin{equation}\label{eq:single-query}
		\Pi_{\leq t+1}{\cal O}_{x,y}\Pi_{t} = {\cal O}_{x,y}\Pi_{t} \;\bot\; {\cal O}_{x,y}\Pi_{\leq t-1} = \Pi_{\leq t}{\cal O}_{x,y}\Pi_{\leq t-1}.
	\end{equation}
	This means that $\Pi_{\leq t+1}{\cal O}_{x,y}\Pi_{t}$ and $\Pi_{\leq t}{\cal O}_{x,y}\Pi_{\leq t-1}$ have orthogonal images and coimages. This orthogonality is preserved after multiplying with $\Lambda_i$ and $\Lambda_{i-1}$ since $\Gamma$ commutes with $\Pi_{\leq t-1}, \Pi_{\leq t}$ and $\Pi_{\leq t+1}$. Hence, we have
	\begin{align*}
		&\norm{\Lambda_{i}\Pi_{\leq t+1}{\cal O}_{x,y}\Pi_{\leq t}\Lambda_{i-1}} = \norm{\Lambda_{i}\left(\Pi_{\leq t}+ \Pi_{t+1}\right){\cal O}_{x,y}\left(\Pi_{\leq t-1}+\Pi_t\right)\Lambda_{i-1}} \\
		& = \norm{\Lambda_{i}\left(\Pi_{\leq t}{\cal O}_{x,y}\Pi_{\leq t-1}+\Pi_{\leq t}{\cal O}_{x,y}\Pi_t+ \Pi_{t+1}{\cal O}_{x,y}\Pi_{\leq t-1}+\Pi_{t+1}{\cal O}_{x,y}\Pi_t\right)\Lambda_{i-1}} \\
		& = \norm{\Lambda_{i}\left(\Pi_{\leq t}{\cal O}_{x,y}\Pi_{\leq t-1}+\Pi_{\leq t+1}{\cal O}_{x,y}\Pi_t\right)\Lambda_{i-1}} \geq \norm{\Lambda_{i}\left(\Pi_{\leq t}{\cal O}_{x,y}\Pi_{\leq t-1}\right)\Lambda_{i-1}}.\qedhere
	\end{align*}
\end{proof}

\section{Reduction from the compressed oracle technique}\label{sec:reduction}

In this section, we present an explicit reduction from the compressed oracle technique to our new MLA method:

\begin{theorem}\label{thm:reduction}
	Fix a finite set $X$ of size $N$ and let $Y = [M-1]_0$. Consider a property ${\cal P} \subseteq (X \times Y)^k$ for some $k \in [M-1]$. Let $\epsilon \in \left(0, 1 - (9 - 4\sqrt{2})\frac{k}{M}\right)$, and fix any problem ${\sf F}: Y^X \rightarrow 2^{\cal P}$. Define the quantities ${\sf MLADV}^{\sf Uniform}_{\epsilon, \frac{2k}{M}}({\sf F})$ and ${\sf COMP}^{\sf Uniform}_{\epsilon}({\sf F})$ as the lower bounds on $Q^{\sf Uniform}_{\epsilon}({\sf F})$ obtained by \cor{mladv} (with $\eta$ set to $\frac{2k}{M}$) and \cor{compressed}, respectively. Then, we have
	$$
	{\sf COMP}^{\sf Uniform}_{\epsilon}({\sf F}) \leq 6 \cdot {\sf MLADV}^{\sf Uniform}_{\epsilon, \frac{2k}{M}}({\sf F}).
	$$
\end{theorem}

Recall from \cor{compressed} that ${\sf COMP}^{\sf Uniform}_{\epsilon}({\sf F})$ is equal to the smallest $T$ satisfying
\begin{equation}\label{eq:minimize}
	\sqrt{1-\epsilon} - \sqrt{\frac{k}{M}} \leq \sum_{t = 1}^T \max\limits_{x \in X,y \in Y}\norm{\mathsf{P}_{{\cal D}_{\leq t} \cap {\cal D}_{{\cal P}}}\mathsf{cO}_{x,y}\mathsf{P}_{{\cal D}_{\leq t-1} \setminus {\cal D}_{{\cal P}}}}.
\end{equation}
We start by removing the compressed oracle $\mathsf{cO}_{x,y}$ in \eq{minimize}. For $t \in [T]_0$, consider the following projections:
\begin{equation}\label{eq:proj-gamma}
	\begin{split}
		\Pi_{1,t} &\coloneqq \mathsf{Comp}^{\dagger}\mathsf{P}_{{\cal D}_{\leq t} \cap {\cal D}_{{\cal P}}}\mathsf{Comp}, \\ 
		\Pi_{0,t}&\coloneqq \mathsf{Comp}^{\dagger}\mathsf{P}_{{\cal D}_{\leq t} \setminus {\cal D}_{{\cal P}}}\mathsf{Comp}.
	\end{split}
\end{equation}
By the definition of $\mathsf{cO}_{x,y}$ from \eq{cfo}, we find that the projections in \eq{proj-gamma} allow us to rewrite the right-hand side of \eq{minimize} as:
\begin{equation}\label{eq:comp-conj}
	\begin{split}
		&\sum_{t = 1}^T \max\limits_{x \in X,y \in Y}\norm{\mathsf{P}_{{\cal D}_{\leq t} \cap {\cal D}_{{\cal P}}}\mathsf{cO}_{x,y}\mathsf{P}_{{\cal D}_{\leq t-1} \setminus {\cal D}_{{\cal P}}}} \\
		&= \sum_{t = 1}^T \max\limits_{x \in X,y \in Y}\norm{\left(\mathsf{Comp}\Pi_{1,t}\mathsf{Comp}^{\dagger}\right)\left(\mathsf{Comp}{\cal O}_{x,y}\mathsf{Comp}^{\dagger}\right)\left(\mathsf{Comp}\Pi_{0,t-1}\mathsf{Comp}^{\dagger}\right)} \nonumber \\
		&= \sum_{t = 1}^T \max\limits_{x \in X,y \in Y}\norm{\Pi_{1,t}{\cal O}_{x,y}\Pi_{0,t-1}}.
	\end{split}
\end{equation}
Hence, by \cor{compressed}, ${\sf COMP}^{\sf Uniform}_{\epsilon}({\sf F})$ is upper bounded by the smallest value of $T$ satisfying
\begin{equation}\label{eq:probboundcomp}
	\sqrt{1-\epsilon} - \sqrt{\frac{k}{M}} \leq \sum_{t = 1}^T \max\limits_{x \in X,y \in Y}\norm{\Pi_{1,t}{\cal O}_{x,y}\Pi_{0,t-1}}.
\end{equation}

Next, we show that for any ${\cal P}  \subseteq (X \times Y)^k$, we can always construct an explicit MLA $\Gamma$ (see \defin{mladv}), with accompanying parameter $\lambda$ and $\eta = \frac{2k}{M}$ satisfying the conditions of \thm{mladv}, such that any $T$ that satisfies
\begin{equation*}
	1+(\lambda-1)(\sqrt{1-\epsilon}-\sqrt{\eta})^2 \leq \prod_{t=1}^{T}\max\limits_{\substack{i \in [\ell-1]_0\\x \in X,y \in Y}}\left( 1 + \frac{\kappa-1}{\sqrt{\kappa}}\norm{\Lambda_{i+1}\Pi_{\leq t}{\cal O}_{x,y}\Pi_{\leq t-1}\Lambda_{i}}\right)^2
\end{equation*}
also satisfies 
\begin{equation}\label{eq:satisfy}
	\sqrt{1-\epsilon} - \sqrt{\frac{k}{M}} \leq \sum_{t = 1}^{6T} \max\limits_{x \in X,y \in Y}\norm{\Pi_{1,t}{\cal O}_{x,y}\Pi_{0,t-1}}.
\end{equation}
This then proves \thm{reduction} by \cor{mladv} and \eq{probboundcomp}.

For $\ell = 1$, we know from \defin{mladv} that any multiplicative ladder adversary matrix has the following form for some $\kappa > 1$:
$$\Gamma = \Lambda_{0} + \kappa \Lambda_{1}.$$
We set the eigenspaces of $\Gamma$ to correspond to the projections $\Lambda_{1} \coloneqq \Pi_{1,N}$ (see \eq{proj-gamma}) and $\Lambda_{0} \coloneqq I - \Lambda_{1}$.
\begin{claim}\label{clm:equal-proj}
	For each $t \in [T]_0$
	\begin{align}
		&\Pi_{\leq t}\Lambda_{0} =\Pi_{0,t},
		&\Lambda_{1}\Pi_{\leq t} = \Pi_{1,t}.
	\end{align}
\end{claim}
\begin{proof}
	Since $\mathsf{P}_{{\cal D}_{\leq N}}$ is the identity on $\mathbb{C}[\left(Y \cup \{\bot\}\right)^X]$, we find that
	\begin{align*}
		\Lambda_{0} = I - \Lambda_{1} &= \mathsf{Comp}^{\dagger}\mathsf{P}_{{\cal D}_{\leq N}}\mathsf{Comp} - \mathsf{Comp}^{\dagger}\mathsf{P}_{{\cal D}_{\leq N} \cap {\cal D}_{{\cal P}}}\mathsf{Comp} \\
		&= \mathsf{Comp}^{\dagger}\mathsf{P}_{{\cal D}_{\leq N} \setminus {\cal D}_{{\cal P}}}\mathsf{Comp}.		
	\end{align*}
	By \eq{space-t-uniform} we know that $\Pi_{\leq t} = \mathsf{Comp}^{\dagger}\mathsf{P}_{{\cal D}_{\leq t}}\mathsf{Comp}$. Together with the commutativity of the projectors onto subsets of ${\cal D}$ (see \eq{proj}), this implies that 
	\begin{align*}
		\Pi_{\leq t}\Lambda_{0} &= \mathsf{Comp}^{\dagger}\mathsf{P}_{{\cal D}_{\leq t}}\mathsf{Comp}\mathsf{Comp}^{\dagger}\mathsf{P}_{{\cal D}_{\leq N} \setminus {\cal D}_{{\cal P}}}\mathsf{Comp} \\  &= \mathsf{Comp}^{\dagger}\mathsf{P}_{{\cal D}_{\leq t} \setminus {\cal D}_{{\cal P}}}\mathsf{Comp}= \Pi_{0,t}.
	\end{align*}
	Similarly we have
	\begin{align*}
		\Lambda_{1}\Pi_{\leq t} &=\mathsf{Comp}^{\dagger}\mathsf{P}_{{\cal D}_{\leq N} \cap {\cal D}_{{\cal P}}}\mathsf{Comp}\mathsf{Comp}^{\dagger}\mathsf{P}_{{\cal D}_{\leq t}}\mathsf{Comp} \\ 
		&= \mathsf{Comp}^{\dagger}\mathsf{P}_{{\cal D}_{\leq t} \cap {\cal D}_{{\cal P}}}\mathsf{Comp} = \Pi_{1,t}.\qedhere
	\end{align*}
\end{proof}

\begin{claim}\label{clm:mladv}
	Let $\Lambda_{1} \coloneqq \Pi_{1,N}$ (see \eq{proj-gamma}), $\Lambda_0 = I - \Lambda_1$ and $\Gamma = \Lambda_0 + \kappa \Lambda_1$ for some constant $\kappa > 1$. Then $\Gamma$ is an MLA matrix as defined in \defin{mladv} with $\ket{{\sf Uniform}}$ as a $1$-eigenvector. 
\end{claim}
\begin{proof}
	It is clear that this construction makes $\Gamma$ positive definite with smallest eigenvalue $1$ and largest eigenvalue $\kappa$. Moreover, since $\ell=1$, we automatically satisfy \eq{ladder} from \defin{mladv}. $\Gamma$ also commutes with every $\Pi_{\leq t}$ due to the commutativity of the projectors onto subsets of ${\cal D}$. Hence, it only rests us to verify that $\ket{{\sf Uniform}} = \frac{1}{\sqrt{M^{N}}} \sum_{f \in Y^X} \ket{f}$ is indeed an eigenvector of $\Gamma$ with eigenvalue~$1$:
	\begin{align*}
		\Lambda_{0}\ket{{\sf Uniform}} &= \ket{{\sf Uniform}} - \Lambda_{1}\ket{{\sf Uniform}} =  \ket{{\sf Uniform}} - \mathsf{Comp}^{\dagger}\mathsf{P}_{{\cal D}_{{\cal P}}}\ket{{{\bot}}}^{\otimes N} =\ket{{\sf Uniform}},	
	\end{align*}
	since the empty database $\ket{{{\bot}}}^{\otimes N}$ can never be an element of ${\cal D}_{{\cal P}}$. 
\end{proof}

Knowing that $\Gamma$ is an MLA with $\ket{{\sf Uniform}}$ as a $1$-eigenvector, we may apply \cor{mladv}. By taking the natural logarithm of both sides, it states
\begin{equation*}
	\ln\left( 1+(\lambda-1)(\sqrt{1-\epsilon}-\sqrt{\eta})^2\right) \leq 2\sum_{t=1}^{T}\ln\left(1 + \max\limits_{x \in X,y \in Y}\frac{\kappa-1}{\sqrt{\kappa}}\norm{\Lambda_{1}\Pi_{\leq t}{\cal O}_{x,y}\Pi_{\leq t-1}\Lambda_{0}}\right).
\end{equation*}
To show that this implies \eq{probboundcomp}, we set $\lambda = \kappa =  1 + (e-1)/\left(\sqrt{1-\epsilon}-\sqrt{\eta}\right)^2$ and multiply both sides of the equation with $\sqrt{1-\epsilon}-\sqrt{\eta}$ to arrive at
\begin{align}\label{eq:prob-bound}
	\sqrt{1-\epsilon}-\sqrt{\eta} &\leq 2\left(\sqrt{1-\epsilon}-\sqrt{\eta}\right)\sum_{t=1}^{T}\ln\left(1 + \max\limits_{x \in X,y \in Y}\frac{\kappa-1}{\sqrt{\kappa}}\norm{\Lambda_{1}\Pi_{\leq t}{\cal O}_{x,y}\Pi_{\leq t-1}\Lambda_{0}}\right) \nonumber \\
	&\leq 2\left(\sqrt{1-\epsilon}-\sqrt{\eta}\right)\frac{\kappa-1}{\sqrt{\kappa}}\sum_{t = 1}^{T}\max\limits_{x \in X,y \in Y}\norm{\Lambda_{1}\Pi_{\leq t}{\cal O}_{x,y}\Pi_{\leq t-1}\Lambda_{0}} \nonumber \\
	&\leq 3\sum_{t = 1}^{T}\max\limits_{x \in X,y \in Y}\norm{\Lambda_{1}\Pi_{\leq t}{\cal O}_{x,y}\Pi_{\leq t-1}\Lambda_{0}}.
\end{align}

To finalise the proof, we show that the choice of $\eta = \frac{2k}{M}$ satisfies the conditions of \thm{mladv}. By our choice of $\Gamma, \lambda, \kappa$, the projection $\Lambda_{{\sf bad}}$ is equal to $\Lambda_{0}$. The proof of the following lemma can be skipped if the reader is familiar with the compressed oracle technique, as the technique is reminiscent to the proof of the lemma in~\cite{zhandry2019record} that links the compressed Fourier oracle to the original oracle.
\begin{lemma}\label{lem:eta-bound}
	Let $\Gamma = \Lambda_{0} + \kappa\Lambda_{1}$ be a multiplicative adversary matrix (see \defin{adv-matrix}) with $\Lambda_{1} =\mathsf{Comp}^{\dagger}\mathsf{P}_{{\cal D}_{{\cal P}}}\mathsf{Comp}$ and  $\Lambda_{0} = I - \Lambda_{1}$. Then for every $z\in {\cal P} \subseteq (X \times Y)^k$ we have 
	$$\norm{F_z\Lambda_{0}}^2 \leq \frac{2k}{M},$$
	where $F_z = \sum\limits_{\substack{f \in Y^X:{\sf F}(f) \ni z}}\proj{f}$.
\end{lemma}
\begin{proof}
	We know that $z$ is of the form $(x_1,y_1),\dots,(x_k,y_k)$. Hence, we have that the projector $F_z$ is equal to $F_z = \bigotimes_{i = 1}^k\proj{y_i}_{\mathcal{I}_{x_i}}$ (and acts as the identity on all other registers of $\mathcal{I}$). $F_z$ is therefore equal to $\mathsf{P}_{{\cal D}_{\{z\}}}$, but the latter acts on $\mathbb{C}[\left(Y \cup \{\bot\}\right)^X]$, whereas the former acts on $\mathbb{C}[Y^X]$. By definition of $\Lambda_{0}$, we find that
	\begin{align*}
		\Lambda_{0} = I - \mathsf{Comp}^{\dagger}\mathsf{P}_{{\cal P}}\mathsf{Comp} \preceq I - \mathsf{Comp}^{\dagger}\mathsf{P}_{{\cal D}_{\{z\}}}\mathsf{Comp}.
	\end{align*}
	Combined with the projection $F_z$ this yields 
	\begin{align}\label{eq:normdiff}
		\norm{F_z\Lambda_{0}} \leq \norm{F_z - F_z\mathsf{Comp}^{\dagger}\mathsf{P}_{{\cal D}_{\{z\}}}\mathsf{Comp}}.
	\end{align}
	
	The projections $F_z$ and $\mathsf{P}_{{\cal D}_{\{z\}}}$ are easier to analyse if we view each $\mathcal{I}_{x_i}$ register in the Fourier basis. If we abuse the equality sign, since both projectors act on slightly different Hilbert spaces, they look as follows in the Fourier basis:
	
	\begin{align}\label{eq:fproj}
		F_z = \mathsf{P}_{{\cal D}_{\{z\}}} = \bigotimes_{i = 1}^k \left(\frac{1}{M}\sum_{\substack{v,w \in Y}} e^{\frac{2\pi \iota}{M}\left(w-v\right) \cdot y_i}\ket{w}\bra{v}_{\mathcal{I}_{x_i}}\right),
	\end{align}
	and hence
	\begin{equation}\label{eq:pproj}
		\begin{split}
			&F_z\mathsf{Comp}^{\dagger}\mathsf{P}_{{\cal D}_{\{z\}}}\mathsf{Comp}\\
			&= \bigotimes_{i = 1}^k \left(\frac{1}{M}\sum_{\substack{v,w \in Y}} e^{\frac{2\pi \iota}{M}\left(w-v\right) \cdot y_i}\ket{w}\bra{v}_{\mathcal{I}_{x_i}}\right)\bigotimes_{i = 1}^k \left(\frac{1}{M}\sum_{\substack{v,w \in (Y\setminus \{0\})}} e^{\frac{2\pi \iota}{M}\left(w-v\right) \cdot y_i}\ket{w}\bra{v}_{\mathcal{I}_{x_i}}\right) \\
			&= \bigotimes_{i = 1}^k \left(\frac{M-1}{M^2}\sum_{\substack{v,w \in Y:\\v \neq 0}} e^{\frac{2\pi \iota}{M}\left(w-v\right) \cdot y_i}\ket{w}\bra{v}_{\mathcal{I}_{x_i}}\right).
		\end{split}
	\end{equation}
	
	We abbreviate $\bm{v} := (v_1,\dots,v_k)$ and similarly introduce $\ket{\bm{v}}_{\mathcal{I}_{\bm{x}}} :=  \bigotimes_{i = 1}^k\ket{v_i}_{\mathcal{I}_{x_i}}$. We also abbreviate 
	$$e^{\frac{2\pi \iota}{M}\left(\bm{w}-\bm{v}\right) \cdot \bm{y}} := \prod_{i=1}^k e^{\frac{2\pi \iota}{M}\left(w_i-v_i\right) \cdot y_i}.$$
	Using this new notation, we can apply both \eq{fproj} and \eq{pproj} to expand the expression $F_z - F_z\mathsf{Comp}^{\dagger}\mathsf{P}_{{\cal D}_{\{z\}}}\mathsf{Comp}$ from \eq{normdiff} as
	\begin{align*}
		&\bigotimes_{i = 1}^k \left(\frac{1}{M}\sum_{\substack{v,w \in Y}} e^{\frac{2\pi \iota}{M}\left(w-v\right) \cdot y_i}\ket{w}\bra{v}_{\mathcal{I}_{x_i}}\right) - \bigotimes_{i = 1}^k \left(\frac{M-1}{M^2}\sum_{\substack{v,w \in Y:\\v \neq 0}} e^{\frac{2\pi \iota}{M}\left(w-v\right) \cdot y_i}\ket{w}\bra{v}_{\mathcal{I}_{x_i}}\right)
		\\
		&= \frac{1}{M^k}\sum_{\substack{\bm{v},\bm{w} \in Y^k}} e^{\frac{2\pi \iota}{M}\left(\bm{w}-\bm{v}\right) \cdot \bm{y}}\ket{\bm{w}}\bra{\bm{v}}_{\mathcal{I}_{\bm{x}}} -  \left(\frac{M-1}{M^2}\right)^k\sum_{\substack{\bm{v},\bm{w} \in Y^k:\\\nexists i: v_i = 0}} e^{\frac{2\pi \iota}{M}\left(\bm{w}-\bm{v}\right) \cdot \bm{y}}\ket{\bm{w}}\bra{\bm{v}}_{\mathcal{I}_{\bm{x}}} \\
		&= \frac{1}{M^k}\sum_{\substack{\bm{v},\bm{w} \in Y^k\\ \exists i: v_i = 0}} e^{\frac{2\pi \iota}{M}\left(\bm{w}-\bm{v}\right) \cdot \bm{y}}\ket{\bm{w}}\bra{\bm{v}}_{\mathcal{I}_{\bm{x}}} +  \left(\frac{1}{M^k} - \left(\frac{M-1}{M^2}\right)^k\right)\sum_{\substack{\bm{v},\bm{w} \in Y^k:\\\nexists i: v_i = 0}} e^{\frac{2\pi \iota}{M}\left(\bm{w}-\bm{v}\right) \cdot \bm{y}}\ket{\bm{w}}\bra{\bm{v}}_{\mathcal{I}_{\bm{x}}}.
	\end{align*}
	We now bound its norm by applying a counting argument on the number of $\bm{v},\bm{w} \in Y^k$ where either one or none of the $v_i$ is equal to $0$:
	\begin{align*}
		&\norm{F_z - F_z\mathsf{Comp}^{\dagger}\mathsf{P}_{{\cal D}_{\{z\}}}\mathsf{Comp}}^2 \\
		&\leq \frac{1}{M^{2k}}\left(M^{2k} -M^k(M-1)^k\right) + \left(\frac{1}{M^{k}} - \left(\frac{M-1}{M^2}\right)^k\right)^2M^k(M-1)^k \nonumber \\
		&= 1 - \left(\frac{M-1}{M}\right)^{k}+ \left(1 - \left(\frac{M-1}{M}\right)^k\right)^{2}\left(\frac{M-1}{M}\right)^k \\
		&=1 - 2\left(1 - \frac{1}{M}\right)^{2k} + \left(1 - \frac{1}{M}\right)^{3k} \leq 1 - \left(1 - \frac{1}{M}\right)^{2k} \leq \frac{2k}{M}.
	\end{align*}
	In the final inequality we have made use of the fact that $M > k \geq 1$, allowing us to apply Bernoulli's inequality: 
	\begin{equation*}
		\left(1 - \frac{1}{M}\right)^{2k} \geq 1 - \frac{2k}{M}.\qedhere
	\end{equation*}
\end{proof}

Knowing that $\frac{2k}{M}$ is a valid value for $\eta$, suppose that $\epsilon \leq 1 - (9-4\sqrt{2})\frac{k}{M}$. Then
$$\sqrt{1-\epsilon} - \left(2\sqrt{2}-1\right)\sqrt{\frac{k}{M}} \geq 0. $$
Together with \clm{equal-proj}, \eq{prob-bound} and \lem{eta-bound}, this means that our MLA matrix $\Gamma$ satisfies \eq{satisfy}, where in the penultimate step we use that $\norm{\Lambda_{1}\Pi_{\leq t}{\cal O}_{x,y}\Pi_{\leq t-1}\Lambda_{0}}$ is monotonically non-decreasing in $t$ (see \fct{mono}): 
\begin{align*}
	&\sqrt{1-\epsilon} - \sqrt{\frac{k}{M}} \leq \sqrt{1-\epsilon} - \sqrt{\frac{k}{M}} + \sqrt{1-\epsilon} - (2\sqrt{2}-1)\sqrt{\frac{k}{M}} = 2\left(\sqrt{1-\epsilon} - \sqrt{\eta}\right) \\
	&\leq \sum_{t = 1}^{6T}\max\limits_{x \in X,y \in Y}\norm{\Lambda_{1}\Pi_{\leq t}{\cal O}_{x,y}\Pi_{\leq t-1}\Lambda_{0}} \leq \sum_{t = 1}^{6T}\max\limits_{x \in X,y \in Y}\norm{\Pi_{1,t}{\cal O}_{x,y}\Pi_{0,t-1}}.
\end{align*}

\section{A strong direct product theorem}\label{sec:dpt}

The machinery of MLA matrices seems a bit overcomplicated compared to what we actually needed in the reduction in \sec{reduction}. Since we only considered multiplicative adversary matrices where $\ell = 1$, we obtain the ``ladder'' property automatically. We will need general $\ell$ however if we want to compute a function ${\sf F}$ on $\ell$ independent instances simultaneously. 

Although it does not seem to fit in the framework of~\cite{chung2021compressed} directly, the compressed oracle framework also has the powerful property of being able to exhibit \textit{strong direct product theorems} (SDPT), as shown in~\cite{liu2019finding,hamoudi2020quantum}. Such a theorem states that if we try to compute ${\sf F}$ on $k$ independent inputs in fewer queries than ${k}$ times the queries needed for a single instance of ${\sf F}$, then our success probability will decrease exponentially in ${k}$. 

It was already shown by~\cite{vspalek2008multiplicative} that the multiplicative adversary method directly satisfies a SDPT. Here we show that a similar proof as in~\cite{ambainis2011symmetry}, which is based on the proof in~\cite{vspalek2008multiplicative}, also holds for the MLA method due to the fact (which we will prove) that the set of MLA matrices is closed under tensor powers. This motivates the study of the MLA method as a simplification of the multiplicative adversary method, since it maintains the property of satisfying a SDPT.

We introduce the following notation for this section: for any problem ${\sf F}:\Func\rightarrow 2^\Sigma$ and integer $k \geq 1$ let ${\sf F}^{(k)}: \Func^k\rightarrow (2^\Sigma)^k$ be defined as
$${\sf F}^{(k)}(h_1,\dots,h_{k}) = ({\sf F}(h_1),\dots,{\sf F}(h_{k})).$$ 
\begin{theorem}\label{thm:dpt}
	For any problem ${\sf F}:\Func\rightarrow2^\Sigma$, input distribution $\delta$ on $\Func$, and fixed $\eta \leq \frac{1}{2}$, let ${\sf MLADV}^\delta_{\epsilon, \eta}({\sf F})$ be the lower bound on $Q_{\epsilon}^\delta({\sf F})$ obtained by \cor{mladv}. Then there exists a constant $c \in (0,1)$ such that for any integer $k > 361$ we have $${\sf MLADV}^{\delta^k}_{\smash{\scriptstyle 1- c^k, \eta^{\frac{2k}{5}}}}({\sf F}^{(k)}) \geq \frac{k}{10} {\sf MLADV}^\delta_{1-\epsilon, \eta}({\sf F}).$$
\end{theorem}

\begin{proof}
	Let $\Gamma, \lambda$ denote the optimal values in \cor{mladv} for a fixed $\eta \leq \frac{1}{2}$. We use these to construct $\Gamma'$ (with eigenspaces denoted by $\Lambda'_j$), $\lambda'$ and $\eta'$ for ${\sf F}^{(k)}$ as follows:
	\begin{align*}
		\Gamma' := \Gamma^{\otimes k}, \lambda' := \lambda^{\frac{k}{10}}, \eta' := \eta^{\frac{2k}{5}}.
	\end{align*}
	This construction yields a positive definite matrix $\Gamma' \in \mathbb{C}^{\Func^k \times \Func^k}$ with smallest eigenvalue $1$ of the form
	\begin{equation}\label{eq:gamma-k}
		\Gamma' = \Gamma^{\otimes k} = \sum_{j=0}^{k\cdot\ell}\kappa^j \Lambda'_j,
	\end{equation}
	where 
	$$\Lambda'_j = \sum_{\substack{i_1,\dots,i_k \in [j]_0:\\i_1+\cdots+i_k = j}}\Lambda_{i_1}\otimes \cdots \otimes \Lambda_{i_k}.$$
	We similarly define
	$$\Pi'_{\leq t} = \sum_{\substack{t_1,\dots,t_k \in [t]_0:\\t_1+\cdots+t_k = t}}\Pi_{\leq t_1}\otimes \cdots \otimes \Pi_{\leq t_k},$$
	where each $\Pi_t := \Pi_{\leq t} - \Pi_{\leq t-1}$ (as in the proof of \fct{mono}). We now set out to show that $\Gamma'$ is of the correct form to use it as an upper bound for ${\sf MLADV}^{\delta^k}_{\smash{\scriptstyle 1 - c^k, \eta^{\frac{2k}{5}}}}({\sf F}^{(k)})$:
	\begin{lemma}\label{lem:mladv-power}
		Let $\Gamma$ be an MLA matrix for ${\sf F}$ (see \defin{mladv}) with $\ket{\delta}$ as a $1$-eigenvector. Then for any non-negative integer $k$, $\Gamma^{\otimes k}$ is an MLA matrix for ${\sf F}^{(k)}$ with $\ket{\delta}^{\otimes k}$ is a $1$-eigenvector .
	\end{lemma}
	\begin{proof}
		By construction, $\Gamma' = \Gamma^{\otimes k}$ is already of the desired form (see \eq{gamma-k} and \defin{mladv}), and since $\ket{\delta}$ is a $1$-eigenvector of $\Gamma$, $\ket{\delta}^{\otimes k}$ is a $1$-eigenvector of $\Gamma'$. Additionally, since $\Gamma$ commutes with each $\Pi_{\leq t}$, it follows that $\Gamma'$ commutes with $\Pi'_{\leq t}$. What remains to verify is that $\Gamma'$ satisfies \eq{ladder}.
		
		Since $\Func^k = (Y^X)^k = (Y^{X'})$, where $X'$ is a set of size $kN$, there exist unique $x$ and $k'$ for every $x' \in X'$ and $y \in Y$ such that 
		$$
		{\cal O}_{x',y} = I^{\otimes k'-1} \otimes {\cal O}_{x,y} \otimes I^{\otimes k-k'},
		$$
		where $I$ is the identity on $\mathbb{C}^{\Func \times \Func}$. Then, since $\Gamma$ commutes with each $\Pi_t$, we can decompose $\norm{\Lambda'_{j'}\Pi'_{\leq t}{\cal O}_{x',y}\Pi'_{\leq t-1}\Lambda'_{j}}$ as
		\begin{equation}\label{eq:norm-k}
			\begin{split}
				&\bigg\|\sum_{\substack{i_1,\dots,i_k \in [j]_0,\\i'_1,\dots,i'_k \in [j']_0:\\i_1+\cdots+i_k = j,\\i'_1+\cdots+i'_k = j'}}
				\sum_{\substack{t_1,\dots,t_k \in [t]_0,\\ t'_1,\dots,t'_k \in [t]_0:\\t_1+\cdots+t_k = t-1,\\t'_1+\cdots+t'_k = t}}
				\underbrace{\Lambda_{i'_1}\Pi_{t'_1}\Pi_{t_1}\Lambda_{i_1}}_{=\delta_{i_1,i'_1}\delta_{t_1,t'_1}\Pi_{t_1}} \otimes \cdots \otimes \Lambda_{i'_{k'}}\Pi_{t'_{k'}}{\cal O}_{x,y}\Pi_{t_k}\Lambda_{i_{k'}} \otimes \cdots \otimes \underbrace{\Lambda_{i'_k}\Pi_{t'_k}\Pi_{t_k}\Lambda_{i_k}}_{=\delta_{i_k,i'_k}\delta_{t_k,t'_k}\Pi_{t_k}}\bigg\| \\
				&\leq \bigg\|\sum_{\substack{i \in [j]_0, i' \in [j']_0:\\i'-i = j'-j}}\Lambda_{i'}\Pi_{\leq t}{\cal O}_{x,y}\Pi_{\leq t-1}\Lambda_{i}\bigg\| 
				\leq \max\limits_{\substack{i \in [j]_0, i' \in [j']_0:\\i'-i = j'-j}}\norm{\Lambda_{i'}\Pi_{\leq t}{\cal O}_{x,y}\Pi_{\leq t-1}\Lambda_{i}}.
			\end{split}    
		\end{equation}
		In the last inequality, we have used the fact that each term $\Lambda_{i'}\Pi_{\leq t}{\cal O}_{x,y}\Pi_{\leq t-1}\Lambda_{i}$ in the sum has orthogonal images and coimages. It now follows from \eq{norm-k} that for every $x' \in X'$, $y \in Y$, $t \leq T$, and $j, j' \in [k\cdot\ell]_0$ with $\abs{j - j'} > 1$, we have
		\begin{align}\label{eq:ladder-k}
			\norm{\Lambda'_{j'}\Pi'_{\leq t}{\cal O}_{x',y}\Pi'_{\leq t-1}\Lambda'_{j}}
			&\leq \max\limits_{\substack{x \in X,i \in [j]_0, i' \in [j']_0:\\i'-i = j'-j}}\norm{\Lambda_{i'}{\cal O}_{x,y}\Lambda_{i}} = 0,
		\end{align}
		since $\Gamma$ is an MLA matrix and thus satisfies \eq{ladder} itself and the fact that $\Gamma$ commutes with both $\Pi_{\leq t-1}$ and $\Pi_{\leq t}$.	
	\end{proof}
	
	To prove \thm{dpt}, we show that any $T$ satisfying 
	\begin{align}\label{eq:mla-plural}
		1 + \left(\lambda'-1\right)(\sqrt{c^k} - \sqrt{\eta'})^2 \leq \prod_{t=1}^{T}\left(1 +\max\limits_{\substack{j \in [k\cdot\ell-1]_0, \\x' \in X^k,y \in Y}}  \frac{\kappa - 1}{\sqrt{\kappa}}\norm{\Lambda'_{j+1}\Pi'_{\leq t}{\cal O}_{'x,y}\Pi'_{\leq t-1}\Lambda'_{j}}\right)^2,
	\end{align}
	also satisfies
	\begin{align}\label{eq:mla-single}
		1 + \left(\lambda-1\right)(\sqrt{1 - \epsilon} - \sqrt{\eta})^2 \leq \prod_{t=1}^{(10/k) T}\left(1 + \max\limits_{\substack{i \in [\ell-1]_0, \\x \in X,y \in Y}} \frac{\kappa - 1}{\sqrt{\kappa}}\norm{\Lambda_{i+1}\Pi_{\leq t}{\cal O}_{x,y}\Pi_{\leq t-1}\Lambda_{i}}\right)^2.
	\end{align}
	The theorem then follows from \cor{mladv}. For the choices of $\lambda'$ and $\eta'$, we know by the assumptions in \thm{mladv} that $\norm{F_z\Lambda_{{\sf bad}}}^2\leq \eta$ for every $z\in\Sigma$, where $\Lambda_{{\sf bad}}$ is the projector onto the eigenspaces of $\Gamma$ corresponding to eigenvalues smaller than $\lambda$ and $F_z = \sum_{\substack{f \in \Func:{\sf F}(f) \ni z}}\proj{f}$. Now let $\Lambda'_{{\sf bad}}$ be the projector onto the eigenspaces of $\Gamma'$ corresponding to eigenvalues smaller than $\lambda'$. Abbreviate $\bm{z} = (z_1,\dots,z_k) \in \Sigma^k$  and define $F_{\bm{z}} = \bigotimes_{j=1}^k F_{z_j}$. Let $V_{{\sf bad}}$ denote the space that $\Lambda_{{\sf bad}}$ projects onto, let $V_{{\sf good}}$ be its orthogonal complement and analogously define $V'_{{\sf bad}}$. By construction of $\Gamma' = \Gamma^{\otimes k}$, we know that $\Lambda'_{{\sf bad}}$ is a subspace of the direct sum of spaces $V_{\bm v}:=\bigotimes_{j = 1}^k V_{v_j}$ where $\bm{v} = (v_1,\dots,v_k) \in \{{\sf good, bad}\}^k$. Since all the eigenvalues of $\Gamma$ are bounded below by $1$ and $V'_{{\sf bad}}$ is the direct sum of all eigenspaces of $\Gamma'$ with eigenvalue smaller than $\lambda' = \lambda^{k/10}$, it must be that the number of $\sf{good}$ subspaces, denoted by $\abs{\bm{v}}$, is at most $k/10$. This means that we can decompose any state $\ket{\phi} \in V'_{{\sf bad}}$ as a product state
	$$ \ket{\phi} = \sum_{\substack{\bm{v} \in \{\textnormal{good, bad}\}^k:\\\abs{\bm{v}} \leq \frac{k}{10}}}\alpha_{\bm{v}}\ket{\phi_{\bm{v}}} = \sum_{\substack{\bm{v} \in \{\textnormal{good, bad}\}^k:\\\abs{\bm{v}} \leq \frac{k}{10}}}\alpha_{\bm{v}}\bigotimes_{j = 1}^k\ket{\phi_{v_j}},$$ 
	where $\ket{\phi_{v_j}} \in V_{v_j}$. It now follows, as in~\cite{vspalek2008multiplicative}, whenever $\eta \leq 1/2$ and $k \geq 361$, that for every $\bm{z} \in \Sigma^k$ there exists $\ket{\phi} \in  V'_{{\sf bad}}$ such that
	\begin{align*}
		\norm{F_{\bm{z}}\Lambda'_{{\sf bad}}}^2 &=\norm{F_{\bm{z}}\ket{\phi}}^2 =  \big\lVert\sum_{\substack{\bm{v} \in \{\textnormal{good, bad}\}^k:\\\abs{\bm{v}} \leq \frac{k}{10}}}\alpha_{\bm{v}}F_{\bm{z}}\ket{\phi_{\bm{v}}}\big\rVert^2 \leq \big\lVert\sum_{\substack{\bm{v} \in \{\textnormal{good, bad}\}^k:\\\abs{\bm{v}} \leq \frac{k}{10}}}\bigotimes_{j=1}^kF_{z_j}\ket{\phi_{v_j}}\big\lVert^2 \\
		&\leq \eta^{\frac{9k}{10}}\sum_{\substack{\bm{v} \in \{\textnormal{good, bad}\}^k:\\\abs{\bm{v}} \leq \frac{k}{10}}} \leq \eta^{\frac{9k}{10}}\sum_{i=0}^{k/10}\binom{k}{i} \leq k\binom{k}{k/10}\eta^{\frac{9k}{10}} \leq k(10e)^{k/10}\eta^{\frac{9k}{10}}.
	\end{align*}
	Under the assumptions of the lemma, we know that $\eta \leq \frac{1}{2}$ and $k \geq 361$, meaning
	$$k(10e)^{k/10}\eta^{\frac{9k}{10}} \leq 2^{k/2}\eta{k/2}\eta^{2k/5} \leq \eta^{2k/5} = \eta'.$$
	
	To finalise the proof, note that for any fixed $\eta < 1/2, \epsilon \in (0,1-\eta)$, and $\lambda > 1$ we have
	$$\frac{1 + \left(\lambda-1\right)(\sqrt{1-\epsilon}-\sqrt{\eta})^2}{\lambda} < \frac{1 + (\lambda-1)}{\lambda} = 1.$$
	Hence, there exists a constant $c \in (0,1)$ such that for every $k \geq 361$ we have
	\begin{equation}\label{eq:bound-lambda}
		\left(\frac{1 + \left(\lambda-1\right)(\sqrt{1-\epsilon}-\sqrt{\eta})^2}{\lambda}\right)^{\frac{k}{10}} + \eta^{\frac{k}{5}} \leq c^{\frac{k}{2}}. 
	\end{equation}
	Therefore, given our choices of $\lambda'$ and $\eta'$, we find
	\begin{equation}\label{eq:succprob-k}
		\begin{split}
			1 + \left(\lambda'-1\right)\left(\sqrt{c^k} - \sqrt{\eta'}\right)^2 &\geq  1 + \left(\lambda'-1\right)\left(\frac{1 + \left(\lambda-1\right)(\sqrt{1-\epsilon}-\sqrt{\eta})^2}{\lambda}\right)^{\frac{k}{10}} \\
			&= 1 + (1- \lambda^{-\frac{k}{10}})\left(1 + \left(\lambda-1\right)(\sqrt{1-\epsilon}-\sqrt{\eta})^2\right)^{\frac{k}{10}}\\
			&\geq \left(1 + \left(\lambda-1\right)(\sqrt{1-\epsilon}-\sqrt{\eta})^2\right)^{\frac{k}{10}},
		\end{split}
	\end{equation}
	where the final inequality is due to the fact that $1 + \left(\lambda-1\right)(\sqrt{1-\epsilon}-\sqrt{\eta})^2 \leq \lambda$ by \eq{bound-lambda}.
	
	We can conclude the theorem by showing that if \eq{mla-plural} holds, then so must \eq{mla-single}, by combining \eq{norm-k} with \eq{ladder-k} and \eq{succprob-k} and the fact that $\norm{\Lambda_{i+1}\Pi_{\leq t}{\cal O}_{x,y}\Pi_{\leq t-1}\Lambda_{i}}$ is monotonically non-decreasing in $t$ (see \fct{mono}):
	\begin{align*}
		&1 + \left(\lambda-1\right)(\sqrt{1-\epsilon}-\sqrt{\eta})^2 \leq \left(1 + \left(\lambda'-1\right)\left(\sqrt{c^k} - \sqrt{\eta'}\right)^2\right)^{\frac{10}{k}}\\
		&\leq \left(\prod_{t=1}^{T}\left(1 + \max\limits_{\substack{j \in [k\cdot\ell-1]_0, \\x' \in X^k,y \in Y}} \frac{\kappa - 1}{\sqrt{\kappa}}\norm{\Lambda'_{j+1}\Pi'_{\leq t}{\cal O}_{x',y}\Pi'_{\leq t-1}\Lambda'_{j}}\right)^2\right)^{\frac{10}{k}} \\
		&\leq \left(\prod_{t=1}^{T}\left(1 + \max\limits_{\substack{i \in [\ell-1]_0, \\x \in X,y \in Y}} \frac{\kappa - 1}{\sqrt{\kappa}}\norm{\Lambda_{i+1}\Pi_{\leq t}{\cal O}_{x,y}\Pi_{\leq t-1}\Lambda_{i}}\right)^2\right)^{\frac{10}{k}} \\
		&\leq \prod_{t=1}^{(10/k) T}\left(1 + \max\limits_{\substack{i \in [\ell-1]_0, \\x \in X,y \in Y}} \frac{\kappa - 1}{\sqrt{\kappa}}\norm{\Lambda_{i+1}\Pi_{\leq t}{\cal O}_{x,y}\Pi_{\leq t-1}\Lambda_{i}}\right)^2 \qedhere.
	\end{align*} 
\end{proof}

\section{Reduction from the polynomial method}\label{sec:reduction-poly}
 
In this section, we show how we can reduce the polynomial method to our new MLA method. Note that in this section, we can revert to the original notion of “success” (see \defin{complex} and \defin{complex-gen}). The polynomial method, due to~\cite{beals2001QLowerBoundPoly}, allows for lower bounding the quantum query complexity of a boolean function ${\sf F}$ via its approximate degree:
\begin{definition}[Approximate degree]\label{defin:deg}
	For any $\epsilon \geq 0$, the \textit{approximate degree} $\widetilde{\emph{deg}}_{\epsilon}({\sf F})$ of a boolean function ${\sf F} : \{0,1\}^n \mapsto \{0,1\}$ is defined as 
	\begin{equation}\label{eq:deg}
		\widetilde{\emph{deg}}_{\epsilon}({\sf F}) = \min_{p} \left\{\emph{deg}({\sf F}): \forall x \in \{0,1\}^n, \abs{p(x) - {\sf F}(x)} \leq \epsilon\right\},
	\end{equation} 
	where the minimum is taken over all $n$-variate polynomials $p: \mathbb{R}^n \mapsto \mathbb{R}$.
\end{definition}

\begin{theorem}[\cite{beals2001QLowerBoundPoly}]\label{theorem:deg}
	For any Boolean function ${\sf F}$, we have $Q_{\epsilon}({\sf F}) \geq \Omega(\widetilde{\emph{deg}}_{\epsilon}({\sf F}))$. 
\end{theorem}

\subsection{A tighter output condition}

Recall from \thm{madv} that our progress measure in the multiplicative adversary framework must satisfy the following condition, from~\cite{vspalek2008multiplicative,ambainis2011symmetry}:
	\begin{equation}\label{eq:item3}
		W^T(\Gamma,{\cal A})\geq 1 + \left(\lambda - 1\right)\left(\sqrt{1-\epsilon} - \sqrt{\eta}\right)^2
	\end{equation}
	whenever ${\cal A}$ has error at most $\epsilon$.
	To prove the reduction, we need to use the fact that the progress measure must satisfy an even stronger condition, due to~\cite{lee2013strong,magnin2015explicit}, which we now describe.
	
	\begin{definition}[(Hadamard product) fidelity]\label{def:fidel}
		The \textit{fidelity} ${\cal F}(\rho, \sigma)$ between two density matrices $\rho$ and $\sigma$ is defined as
		$${\cal F}(\rho, \sigma) := \emph{\Tr}\left[\sqrt{\sqrt{\rho}\sigma\sqrt{\rho}}\right].$$
		The \textit{Hadamard product fidelity} ${\cal F}(\rho, \sigma)$ (introduced in~\cite{magnin2015explicit}) between two Gram matrices $A$ and $B$ is defined as
		$${\cal F}_H(A,B) := \min\limits_{\ket{u}: \norm{\ket{u}} = 1}{\cal F}\left(A \circ \proj{u},B \circ \proj{u}\right),$$
		where $\circ$ denotes the Hadamard (entrywise) product.
	\end{definition}
	
	Let $M$ be the Gram matrix corresponding to our function ${\sf F}$, i.e. 
	$$M = \sum_{z \in \Sigma} \sum_{\substack{f,f' \in \Func:{\sf F}(f) = {\sf F}(f') = z}}\ket{f}\bra{f'}.$$
	Then in~\cite{lee2013strong,magnin2015explicit} it is shown that the condition 
	\begin{equation}\label{eq:output}
		W^T(\Gamma, {\cal A}) \geq \min_N \left\{\Tr[\Gamma N]: {\cal F}_H(N,M) \geq \sqrt{1-\epsilon}, N \succeq 0, N \circ I = I \right\}
	\end{equation}
	must be satisfied when $\Gamma$ is as in \thm{madv}, for any quantum algorithm $\cal A$ that solves ${\sf F}$ on input $\ket{\delta}$ with success probability at least $1-\epsilon$. 	
	This output condition is stronger than the one from \eq{item3}:
	\begin{fact}\label{fct:output}
		Let $\Gamma$ be a multiplicative adversary matrix for a problem ${\sf F}: \Func \rightarrow 2^\Sigma$ with Gram matrix $M$ and let $\lambda$ satisfy the constraints of \thm{madv}. Let $\Lambda_{{\sf bad}}$ be the projector onto the eigenspaces of $\Gamma$ corresponding to eigenvalues smaller than $\lambda$ and let $\eta\leq 1-\epsilon$ be a positive constant such that $\norm{F_z\Lambda_{{\sf bad}}}^2\leq \eta$ for every $z\in \Sigma$, where $F_z = \smash{\sum_{f \in \Func: z = {\sf F}(f)}}\proj{f}$. 
		Then for every gram matrix $N$ s.t. ${\cal F}_H(N,M) \geq \sqrt{1-\epsilon}$, we have 
		$$\emph{\Tr}[\Gamma N] \geq 1 + \left(\lambda - 1\right)\left(\sqrt{1-\epsilon} - \sqrt{\eta}\right)^2.$$
	\end{fact}
	\noindent The proof of this fact can be found in \appx{output} and was communicated to us by J\'er\'emie Roland.

	This stronger output condition was used in~\cite{lee2013strong} to exhibit an SDPT for quantum query complexity and in~\cite{magnin2015explicit} for the reduction from the polynomial method to the multiplicative adversary method. However, due to its abstract phrasing, it is less suited to prove explicit lower bounds. It is straightforward to reprove our SDPT from \thm{dpt} for this stronger output condition, following the same argument as in~\cite{magnin2015explicit}, by applying \lem{mladv-power}.
	
	Under the output condition from \eq{output}, we obtain the following strengthening of \cor{mladv}:
	\begin{corollary}\label{cor:mladv2}
		For any $\epsilon\in (0,1]$, problem ${\sf F}:\Func\rightarrow\Sigma$ with Gram matrix $M$, the $\epsilon$-error quantum query complexity $Q_{\epsilon}({\sf F})$ is lower bounded by the smallest $T$ such that
		\begin{align*}
			\min_{N: {\cal F}_H(N,M) \geq \sqrt{1-\epsilon}, N \succeq 0, N \circ I = I}\emph{\Tr}[\Gamma N] \leq \min_{\Gamma}\prod_{t=1}^{T}\bigg(1 + \max\limits_{\substack{i \in [\ell-1]_0, \\x \in X,y \in Y}}
			\frac{\kappa-1}{\sqrt{\kappa}}\norm{\Lambda_{i+1}\Pi_{\leq t}{\cal O}_{x,y}\Pi_{\leq t-1}\Lambda_{i}}\bigg)^2.
		\end{align*}
	\end{corollary}

\subsection{The reduction}

We now formally prove our reduction, which takes the following form: 

\begin{theorem}\label{thm:reduction-poly}
	Fix any $\epsilon\in (0,1]$ and problem ${\sf F}:\{0,1\}^n\rightarrow\{0,1\}$. Let ${\sf MLADV}_{\epsilon}({\sf F})$ be the lower bound on $Q_{\epsilon}({\sf F})$ obtained by \cor{mladv2}. Then, we have
	$$
	\widetilde{\emph{deg}}_{\epsilon}({\sf F}) \leq 4 \cdot {\sf MLADV}_{\epsilon}({\sf F}).
	$$
\end{theorem}
\begin{proof}
	
Recall from \fct{size} and \lem{space} that ${\sf Space}_t({{\sf Uniform}})$ is supported on vectors in the Fourier basis of the form $\ket{\hat{f}}$, where  
$$ {f} = {y_1}\cdot\delta_{x_1} + \cdots + {y_s}\cdot\delta_{x_s}, $$
for some $x_1,\dots,x_s \in X$, $y_1,\dots y_s \in Y$, and $s \in [t]_0$. In the case of Boolean functions, where $X = [n]$ and $Y = \{0,1\}$, we can encode $f$ as an $n$-bit string $S$ in the usual way -- meaning, in this case, the $i$-th bit of $S$ is equal to $1$ if there exists an index $j \in [s]$ such that $x_j = i$ and $y_j = 1$, and $0$ otherwise. Hence, we find in this case that ${\sf Space}_t({{\sf Uniform}})$ is supported on the vectors
$$\ket{\chi_S} := \frac{1}{\sqrt{2^n}}\sum_{f \in \{0,1\}^n}(-1)^{S \cdot f}\ket{f},$$
where $S$ is an $n$-bit string of Hamming weight at most $t$. 

For our MLA matrix $\Gamma$ we choose
\begin{equation}\label{eq:ladder-poly}
	\Gamma = \sum_{S \in \{0,1\}^n} \kappa^{|S|} \proj{\chi_S}.
\end{equation}
It is clear that this construction makes $\Gamma$ positive definite with smallest eigenvalue $1$ and corresponding $1$-eigenvector $\ket{{\sf Uniform}} =  \frac{1}{\sqrt{2^n}}\sum_{f \in \{0,1\}^n}\ket{f}$. Additionally, since the $\ket{\hat{f}}$ and thus also the $\ket{\chi_S}$ form an orthogonal basis, we have that
$$\Lambda_i =  \sum_{S \in \{0,1\}^n: |S| = i} \proj{\chi_S} =\Pi_{\leq i} - \Pi_{\leq i-1} = \Pi_i.$$
So not only does $\Gamma$ commute with $\Pi_{\leq t}$ for every $t \in [n]_0$, but due to \eq{one-query} it also satisfies \eq{ladder}: 
$$\norm{\Lambda_{i'} \mathcal{O}_{x,y} \Lambda_i} =  \norm{\Pi_{i'} \mathcal{O}_{x,y} \Pi_i} = 0, \quad \text{if } \abs{i' - i} > 1.$$

The rest of our reduction relies on the reduction in~\cite{magnin2015explicit} from the polynomial method to the multiplicative adversary method. They employ the same choice of multiplicative adversary matrix, which we have just shown to be an MLA matrix, and they show the following:
\begin{fact}[Lemma 16 in~\cite{magnin2015explicit}]
	Let $\Gamma$ be the multiplicative adversary matrix from \eq{ladder-poly} and let ${\sf F} : \{0,1\}^n \mapsto \{0,1\}$ be a Boolean function. Then
	$$ \min_{N: {\cal F}_H(N,M) \geq \sqrt{1-\epsilon}, N \succeq 0, N \circ I = I}\emph{\Tr}[\Gamma N] \geq \frac{\kappa^{\widetilde{\emph{deg}}_{\epsilon}({\sf F})}\epsilon^2}{2^{2n}}.$$
\end{fact}
Plugging this into \cor{mladv2}, we find that ${\sf MLADV}_{\epsilon}({\sf F})$ is the smallest value of $T$ satisfying
$$\frac{\kappa^{\widetilde{\textnormal{deg}}_{\epsilon}({\sf F})}\epsilon^2}{2^{2n}} \leq \min_{\Gamma}\prod_{t=1}^{T}\bigg(1 + \max\limits_{\substack{i \in [\ell-1]_0, \\x \in X,y \in Y}}
\frac{\kappa-1}{\sqrt{\kappa}}\norm{\Lambda_{i+1}\Pi_{\leq t}{\cal O}_{x,y}\Pi_{\leq t-1}\Lambda_{i}}\bigg)^2 \leq \left(1+\frac{\kappa-1}{\sqrt{\kappa}}\right)^{2T}.$$
Taking the logarithm (base $2$) of both sides and plugging in $\kappa = 2^{4(n-\log(\epsilon))}$, we conclude that
\begin{align*}
	{\sf MLADV}_{\epsilon}({\sf F}) \geq \frac{\widetilde{\textnormal{deg}}_{\epsilon}({\sf F})\log(\kappa)}{2\log(1 + \frac{\kappa-1}{\sqrt{\kappa}})} - \frac{n-\log(\epsilon)}{\log(1 + \frac{\kappa-1}{\sqrt{\kappa}})} \geq \frac{\widetilde{\textnormal{deg}}_{\epsilon}({\sf F})}{2} - \frac{n-\log(\epsilon)}{\log(\kappa)} = \frac{\widetilde{\textnormal{deg}}_{\epsilon}({\sf F})}{2} - \frac{1}{4} \geq \frac{\widetilde{\textnormal{deg}}_{\epsilon}({\sf F})}{4}.
\end{align*}     
\end{proof}

\section{Inverting permutations}\label{sec:perm}

In this section, we show that the approach in~\cite{rosmanis2021tight} to generalise the compressed oracle framework to permutations is also captured by the multiplicative ladder adversary (MLA) method. Since in the setting of~\cite{rosmanis2021tight} we are working with random permutations, rather than random functions, we consider ${\sf Perm}$: the set of all permutations from $X$ to $X$, where $X = [N-1]_0$. Our objective is to find the unique preimage of $0$ under a permutation $f$, meaning ${\sf F}: {\sf Perm} \rightarrow X$, where ${\sf F}(f) = x$ if and only if $f(x) = 0$. Like in the previous section, we revert back to the original notion of “success” (see \defin{complex} and \defin{complex-gen}). We aim to apply \cor{mladv} to recover the following result:
\begin{theorem}[Corollary 5 in~\cite{rosmanis2021tight}]\label{thm:perm}
	Let ${\sf Perm}$ be the set of all permutations from $X$ to $X$, where $X = [N-1]_0$ and let ${\sf F}: {\sf Perm} \rightarrow X$, where ${\sf F}(f) = f^{-1}(0)$. Any $T$-query quantum algorithm ${\cal A}$ successfully outputs ${\sf F}(f)$ when $f$ is drawn uniformly from ${\sf Perm}$ with success probability at most $\left(1 + 2\sqrt{2}T\right)^2/\left(N-4T\right)$. 
\end{theorem}

We apply \cor{mladv} by constructing an MLA matrix $\Gamma$ from the constructions in~\cite{rosmanis2021tight}. In the permutation case, the states that make up ${\sf Space}_t(\ket{\delta})$ (see \eq{vstate}) are (for any $t \in [N]_0$):
\begin{equation*}
	\ket{v_{x_1,\dots,x_t}^{y_1,\dots,y_t}} \coloneqq \frac{1}{\sqrt{(N-t)!}}\sum_{\substack{f \in {\sf Perm}\\ \forall i \in [t]: f(x_i) = y_i}}\ket{f}.
\end{equation*}
Each such state can be interpreted as the database $\ket{D}$ from \sec{comp-oracle}, where $D$ contains the input-output pairs $(x_1,y_1),\dots,(x_t,y_t)$. 
In~\cite{rosmanis2021tight}, the span of these states is denoted by $A_t$:
$$ A_t \coloneqq {\sf Space}_t(\ket{\delta}) =  \textrm{span}\left\{\ket{v_{x_1,\dots,x_t}^{y_1,\dots,y_t}}: ((x_1,y_1),\dots,(x_t,y_t)) \in (X \times X)^t \right\}.$$
The second space introduced in~\cite{rosmanis2021tight}, where $t \in [N]$, is
\begin{equation}
	B_t \coloneqq \textrm{span}\left\{\ket{v_{x_1,\dots,x_t}^{0,\dots,y_t}}: ((x_1,0),(x_2,y_2),\dots,(x_t,y_t)) \in (X \times X)^t \right\} \subseteq A_t,
\end{equation}
where a preimage of zero is captured in the database. We have already seen in \eq{decomp-x} that $A_t \subseteq A_{t+1}$. Instead of summing over the different possible $y$ values of the new input-output pair, we can also sum over the possible $x$ values:
$$\ket{v_{x_1,\dots,x_t}^{y_1,\dots,y_t}} \coloneqq \sqrt{N-(k+1)}\sum_{x \in X \setminus \{x_1,\dots,x_t\}}\ket{v_{x_1,\dots,x_t,x}^{y_1,\dots,y_t,y}},$$
where $y$ is any fixed element in $Y \setminus \{y_1,\dots,y_t\}$. By choosing $y = 0$, we actually obtain for every $t \in [N]$ that
\begin{equation}\label{eq:subset}
	A_{t-1} \subseteq B_{t} \subseteq A_{t}.
\end{equation}
With these spaces, we can construct our MLA matrix $\Gamma$. Although it seems reasonable to let the eigenspaces of $\Gamma$ correspond to the spaces $A_t$ and $B_t$, \eq{subset} shows that these spaces are not orthogonal. We address this by introducing the projectors $\widehat{\Pi}_{1,t}$ and $\widehat{\Pi}_{0,t}$, which project onto $\bigoplus_{i = 1}^t \left(B_i \cap \left(A_{i-1}\right)^{\bot}\right)$ and $\bigoplus_{i = 1}^t\left( A_i \cap \left(B_i\right)^{\bot}\right)$, respectively.  In~\cite{rosmanis2021tight}, these projectors are called $\widehat{\Pi}_{t}^{\textnormal{high}}$ and $\widehat{\Pi}_{t}^{\textnormal{low}}$. To understand the intuition as to why these spaces are considered, we refer the reader to~\cite{rosmanis2021tight}. For now, we can think of the projectors $\widehat{\Pi}_{1,t}$ and $\widehat{\Pi}_{0,t}$ as the permutation counterparts of $\Pi_{1,t}$ and $\Pi_{0,t}$ that we defined in \eq{proj-gamma}. Once again, our MLA matrix will be of the form $\Gamma = \Lambda_{0} + \kappa\Lambda_{1}$, where we set $\Lambda_{1} = \widehat{\Pi}_{1,N}$, and accordingly set $\Lambda_{0} = I - \Lambda_{1}$:

\begin{claim}\label{clm:mladv-perm}
	Let $\Lambda_{1} \coloneqq \widehat{\Pi}_{1,N}$, which projects onto $\bigoplus_{i = 1}^N \left(B_i \cap \left(A_{i-1}\right)^{\bot}\right)$, let $\Lambda_0 = I - \Lambda_1$ and $\Gamma = \Lambda_0 + \kappa \Lambda_1$ for some constant $\kappa > 1$. Then $\Gamma$ is an MLA matrix as defined in \defin{mladv} with $\ket{\delta}$ as a $1$-eigenvector. 
\end{claim}
\begin{proof}
	It is clear that this construction makes $\Gamma$ positive definite with smallest eigenvalue $1$ and largest eigenvalue $\kappa$. Additionally, by \eq{subset} each component of the direct sum $\bigoplus_{i = 1}^N \left(B_i \cap \left(A_{i-1}\right)^{\bot}\right)$ is a subspace of $A_{N}$, meaning $\Lambda_1$ (and hence also $\Gamma$ by construction) commutes with $\Pi_{\leq N}$ (which projects onto $A_N$) and hence also with every $\Pi_{\leq t} \preceq \Pi_{\leq N}$. Lastly, since $\ell=1$, we automatically satisfy \eq{ladder} from \defin{mladv}, meaning we only need to verify that $\ket{\delta}= \frac{1}{\sqrt{N!}}\sum_{\substack{f \in {\sf Perm}}}\ket{f}$ is indeed an eigenvector of $\Gamma$ with eigenvalue $1$. Recall from \eq{subset} that $A_{t-1} \subseteq A_{t}$ for $t \in [N]$ . In particular, this means that $\ket{\delta} \in A_0$ is orthogonal to each $\left(A_{t}\right)^{\bot}$ for $t \in [N]_0$ and therefore in particular also to the direct sum $\bigoplus_{t = 1}^N B_t \cap \left(A_{t-1}\right)^{\bot}$, meaning 
	\begin{equation*}
		\Gamma\ket{\delta} = \ket{\delta} - \widehat{\Pi}_{1,N}\ket{\delta} = \ket{\delta}. \qedhere
	\end{equation*}
\end{proof}
By choosing $\lambda = \kappa = 1 + (e-1)/\left(\sqrt{1-\epsilon}-\sqrt{\eta}\right)^2$ (as seen in \eq{prob-bound}), \cor{mladv} now tells us that $Q_{\epsilon}({\sf F})$ is lower bounded by the smallest $T$ satisfying
\begin{equation}\label{eq:prob-bound-perm}
	\sqrt{1-\epsilon}-\sqrt{\eta} \leq 2\sqrt{2}\sum_{t = 1}^{T}\max\limits_{x \in X,y \in Y}\norm{\Lambda_{1}\Pi_{\leq t}{\cal O}_{x,y}\Pi_{\leq t-1}\Lambda_{0}}.
\end{equation}
To be able to continue, we first need to show that \clm{equal-proj} also holds in the case of permutations:
\begin{claim}\label{clm:perm-proj}
	For each $t \in [N]_0$, we have
	\begin{align}
		&\Pi_{\leq t}\Lambda_{0} = \widehat{\Pi}_{0,t},
		&\Lambda_{1}\Pi_{\leq t} = \widehat{\Pi}_{1,t}.
	\end{align}
\end{claim}
\begin{proof}
	Both parts of the claim follow from the fact that $A_{t-1} \subseteq B_t \subseteq A_t$ (see \eq{subset}): Starting with $\Lambda_1$, we know that $\Lambda_{1}\Pi_{\leq t}$ projects onto
	$$A_{t} \cap \bigoplus_{i = 1}^N \left(B_i \cap \left(A_{i-1}\right)^{\bot}\right) = A_{t} \cap \bigoplus_{i = 1}^t \left(B_i \cap \left(A_{i-1}\right)^{\bot}\right) = \bigoplus_{i = 1}^t \left(B_i \cap \left(A_{i-1}\right)^{\bot}\right),$$
	which is the space that $\widehat{\Pi}_{1,t}$ projects onto. Similarly, we obtain that $\Pi_{\leq t}\Lambda_{0}$ projects onto
	$$A_{t} ~\cap~ \left(\bigoplus_{i = 1}^N\left(B_i \cap \left(A_{i-1}\right)^{\bot}\right)\right)^{\bot} =  A_t ~\cap~ \bigcap_{i = 1}^N \left(A_{i-1} \cup \left(B_i\right)^{\bot}\right) = \bigoplus_{i = 1}^t \left(A_i \cap \left(B_{i-1}\right)^{\bot}\right),$$
	which is the space that $\widehat{\Pi}_{0,t}$ projects onto.
\end{proof}

By \clm{perm-proj}, we can relate the right-hand side of \eq{prob-bound-perm} to the projectors $\widehat{\Pi}_{0,t}, \widehat{\Pi}_{1,t}$ via the inequality
\begin{equation}\label{eq:prob-bound-proj}
	\begin{split}
		\norm{\Lambda_{1}\Pi_{\leq t}{\cal O}_{x,y}\Pi_{\leq t-1}\Lambda_{0}} &= \norm{\widehat{\Pi}_{1,t}{\cal O}_{x,y}\widehat{\Pi}_{0,t-1}}.
	\end{split}
\end{equation}
It is shown in Claim $11$ and Claim $12$ in~\cite{rosmanis2021tight} that 
$$ \norm{\widehat{\Pi}_{1,t}{\cal O}_{x,y}\widehat{\Pi}_{0,t-1}} \leq \frac{2\sqrt{2}}{\sqrt{N-4t}}.$$
By plugging this into \eq{prob-bound-perm}, together with \eq{prob-bound-proj}, we obtain
\begin{equation}\label{eq:prob-bound-perm-eta}
	\sqrt{1-\epsilon}-\sqrt{\eta} \leq 2\sqrt{2}\sum_{t = 1}^{T}\max\limits_{x \in X,y \in Y}\frac{2\sqrt{2}}{\sqrt{N-4t}} \leq \frac{8T}{\sqrt{N-4T}}.
\end{equation}
The last step in proving \thm{perm} consists of finding a valid value for $\eta$, meaning we have to bound $\norm{F_z\Lambda_{\sf bad}}$. By construction of $\Gamma$ and our choice of $\lambda$, the projection $\Lambda_{{\sf bad}}$ is equal to $\Lambda_{0}$. The final piece of the puzzle can again be found in~\cite{rosmanis2021tight}, this time in Claim $10$, where it is shown that
$$ \norm{F_z\Lambda_{0}} \leq \frac{1}{\sqrt{N-2T}}.$$
Combining this with \eq{prob-bound-perm-eta}, results in an upper bound on our success probability of 
$$ \left(\frac{8T}{\sqrt{N-4T}} + \frac{1}{\sqrt{N-2T}}\right)^2 \leq \frac{\left(1+8T\right)^2}{N-4T},$$
which recovers \thm{perm} up to a constant factors.

\subsection*{Acknowledgements}

We thank Arne Darras, Sander Gribling and J\'er\'emie Roland for helpful discussions and we thank J\'er\'emie Roland additionally for the proof of \fct{output}.

\newpage
\bibliographystyle{alpha}
\bibliography{refs}
\newpage

\begin{appendix}
	
	\tocless\section{Proof of \fct{output}}\label{appx:output}  
	
	Suppose that ${\cal F}_H(N,M) \geq \sqrt{1-\epsilon}$. From \defin{fidel}, we know there exists (a normalised) $\ket{u} = \sum_{f \in \Func} u_f \ket{f}$ such that 
	\begin{equation*}
		{\cal F}_H(N,M) = {\cal F}\left(N \circ \proj{u},M \circ \proj{u}\right).
	\end{equation*}
	
	Let ${\cal W}_O$ denote the workspace register  containing the output $z\in \Sigma$. By Claim 3.8 in~\cite{lee2013strong}, there exist states $\ket{\psi_f} \in \mathbb{C}^{\Sigma}$ for $f \in \Func$ such that we have both $N = \sum_{f,f' \in \Func}\braket{\psi_{f'}}{\psi_f}\ket{f}\bra{f'}$ and $\Re(\braket{\psi_f}{{\sf F}(f)}) \geq \sqrt{1-\epsilon}$ for every $f \in \Func$. By letting
	\begin{align*}
		&\ket{\Psi} = \sum_{f \in \Func}u_f \ket{\psi_f}_{{\cal W}_O}\ket{f}_{\cal I},\\
		&\ket{\Phi} = \sum_{f \in \Func}u_f \ket{{\sf F}(f)}_{{\cal W}_O}\ket{f}_{\cal I} = F_z\ket{\Phi},
	\end{align*}
	we find that
	\begin{equation}\label{eq:succ2}
	\begin{split} 
		\abs{\braket{\Psi}{\Phi}} \geq \Re(\braket{\Psi}{\Phi}) = \sum_f \abs{u_f}^2 \Re(\braket{\psi_f}{{\sf F}(f)}) \geq \sqrt{1-\epsilon}.
	\end{split}
	\end{equation}

	The rest of the proof will now closely resemble the proof of \lem{item3}. Define $\Lambda_{\sf{good}} := I - \Lambda_{\sf{bad}}$ as the projector onto the orthogonal complement of the bad subspace, which we call the good subspace. Using these projectors, we decompose $\ket{\Psi} = \sqrt{1 - \beta} \ket{\Psi_{\sf{bad}}} + \sqrt{\beta} \ket{\Psi_{\sf{good}}}$, where
	$$
	\ket{\Psi_{\sf{bad}}} = \frac{(I_{{\cal W}_O} \otimes \Lambda_{\sf{bad}}) \ket{\Psi}}{\|(I_{{\cal W}_O} \otimes \Lambda_{\sf{bad}}) \ket{\Psi}\|}, ~~ \ket{\Psi_{\sf{good}}} = \frac{(I_{{\cal W}_O} \otimes \Lambda_{\sf{good}}) \ket{\Psi}}{\|(I_{{\cal W}_O} \otimes \Lambda_{\sf{good}}) \ket{\Psi}\|}, \text{ and } ~~ \beta = \norm{((I_{{\cal W}_O} \otimes \Lambda_{\sf{good}}) \ket{\Psi}}^2.
	$$
	
	For the ``good'' component, we can use the trivial bound $\abs{\braket{\Psi_{\sf{good}}}{\Phi}} \leq 1$. For the ``bad'' component, we bound it by
	$$
	\abs{\braket{\Psi_{\sf{bad}}}{\Phi}} \leq \max_{z \in \Sigma} \norm{F_z \Lambda_{\sf{bad}}} \leq \sqrt{\eta}.
	$$
	Combining this with \eq{succ2} yields
	$$\sqrt{1-\epsilon} \leq \abs{\braket{\Psi}{\Phi}}  \leq \sqrt{1-\beta}\abs{\braket{\Psi_{\sf{bad}}}{\Phi}} + \sqrt{\beta}\abs{\braket{\Psi_{\sf{good}}}{\Phi}} \leq \sqrt{\eta} + \sqrt{\beta},$$
	which we can rearrange to obtain $\beta \geq \left(\sqrt{1 - \epsilon} - \sqrt{\eta}\right)^2$. This allows us to conclude that
	\begin{align*}
		\Tr(\Gamma N) &\geq \Tr(\lambda \Lambda_{\sf{good}} N) + \Tr(\Lambda_{\sf{bad}} N)\geq \lambda \beta + (1 - \beta) \geq 1 + (\lambda - 1) \left(\sqrt{1 - \epsilon} - \sqrt{\eta}\right)^2. \qedhere
	\end{align*}

\end{appendix}

\end{document}